\documentclass[11pt]{article}

\usepackage{bm, commath}
\usepackage{natbib}
\usepackage{caption}
\usepackage{graphicx}
\usepackage{subfig}
\usepackage{amsmath, amsfonts, amsthm}
\usepackage{float}
\usepackage{booktabs,siunitx}
\usepackage{url}
\usepackage{multirow}
\usepackage{amssymb}
\usepackage{blkarray}
\usepackage{bbm}
\usepackage{multicol}
\usepackage{authblk}

\usepackage{listings}
\usepackage{bbm}
\usepackage{textcomp}
\usepackage[margin=1in]{geometry}
\usepackage{authblk}
\usepackage[ruled]{algorithm2e}
\SetKwInput{KwParam}{Parameter}
\SetAlgoCaptionLayout{centerline}

\usepackage{tikz}
\usepackage{array}
\newcolumntype{P}[1]{>{\centering\arraybackslash}p{#1}}
\usetikzlibrary{shapes,decorations,arrows,calc,arrows.meta,fit,positioning}
\tikzset{
    -Latex,auto,node distance =1 cm and 1 cm,semithick,
    state/.style ={ellipse, draw, minimum width = 0.7 cm},
    point/.style = {circle, draw, inner sep=0.04cm,fill,node contents={}},
    bidirected/.style={Latex-Latex,dashed},
    el/.style = {inner sep=2pt, align=left, sloped}
}

\usepackage{sectsty}
\usepackage[onehalfspacing]{setspace}
\usepackage{comment}

\newcommand{\indep}{\rotatebox[origin=c]{90}{$\models$}}

\newtheorem{theorem}{Theorem}

\newtheorem{corollary}{Corollary}

\newtheorem{assumption}{Assumption}
\newtheorem*{assumption*}{Assumption}

\newtheorem{proposition}{Proposition}

\newtheorem{lemma}{Lemma}

\usepackage{mathtools}

\usepackage{xcolor}

\makeatletter
\renewcommand{\algocf@captiontext}[2]{#1\algocf@typo. \AlCapFnt{}#2} 
\def\@algocf@capt@plain{top}
\renewcommand{\algocf@makecaption}[2]{%
  \addtolength{\hsize}{\algomargin}%
  \sbox\@tempboxa{\algocf@captiontext{#1}{#2}}%
  \ifdim\wd\@tempboxa >\hsize
  \hskip .5\algomargin%
  \parbox[t]{\hsize}{\algocf@captiontext{#1}{#2}}
  \else%
  \global\@minipagefalse%
  \hbox to\hsize{\box\@tempboxa}
  \fi%
  \addtolength{\hsize}{-\algomargin}%
}
\makeatother

\begin{document}

\sectionfont{\bfseries\large\sffamily}%

\subsectionfont{\bfseries\sffamily\normalsize}%




\title{A Differential Effect Approach to Partial Identification of Treatment Effects}

\author[1]{Kan Chen  }
\author[2]{Jeffrey Zhang}
\author[3]{Bingkai Wang }
\author[4]{Dylan S. Small \thanks{Email: {\tt dsmall@wharton.upenn.edu.}}}

\affil[1]{Department of Biostatistics, T.H Chan School of Public Health, Harvard University, Boston, Massachusetts, U.S.A.}
\affil[3]{Department of Biostatistics, University of Michigan, Ann Arbor, Michigan, U.S.A}
\affil[2,4]{Department of Statistics and Data Science, The Wharton School, University of Pennsylvania, Philadelphia, Pennsylvania, U.S.A.}

\date{}

\maketitle

\noindent
\textsf{{\bf Abstract}:   We consider identification and inference for the average treatment effect and heterogeneous treatment effect conditional on observable covariates in the presence of unmeasured confounding. Since point identification of these treatment effects is not achievable without strong assumptions, we obtain bounds on these treatment effects by leveraging differential effects, a tool that allows for using a second treatment to learn the effect of the first treatment. The differential effect is the effect of using one treatment in lieu of the other.  We provide conditions under which differential treatment effects can be used to point identify or partially identify treatment effects. Under these conditions, we develop a flexible and easy-to-implement semi-parametric framework to estimate bounds and leverage two two-stage approaches to conduct statistical inference on effects of interest.  The proposed method is examined through a simulation study and a case study that investigates the effect of smoking on the blood level of cadmium using the National Health and Nutrition Examination Survey.}%

\vspace{0.3 cm}

\section{Introduction}

\subsection{Background: Identification of Treatment Effects }

Estimating treatment effects is of great importance for many applications with observational data. In this paper, we consider the problem of identification of the average treatment effect (ATE) and the conditional average treatment effect (CATE). Existing methods provide point identification for the ATE \citep{rubin1974estimating,imbens2015causal,chen2022covariate} with either a parametric or non-parametric approach under the \emph{No Unmeasured Confounding} assumption: $\{Y(1), Y(0)\} \indep Z_1 | X$,
 where $\indep$ means independence, $Z_1 \in \{0,1\}$ is a binary treatment assignment, $X \in \mathbb{R}^l, l \geq 1$ is a set of observed covariates, and $Y(1) \in \mathbb{R}$ represents the potential outcome under intervention while $Y(0) \in \mathbb{R}$ represents the potential outcome under control. The assumption holds when all confounding factors are observed or observed covariates $X$ account for all dependence between the potential outcomes and the treatment assignment. Point identification of the CATE under strong ignorability also has recently been studied by a wide range of literature \citep{hill2011bayesian,abrevaya2015estimating, athey2016recursive, wager2018estimation,kunzel2019metalearners, nie2021quasi}.

However, the no unmeasured confounding assumption rules out unobserved confounders, which can be restrictive and unrealistic.
Therefore, increasing attention has been paid to the identification of ATE and CATE in the presence of unmeasured confounding via instrumental variable (IV) estimation \citep{angrist1996identification,tan2006regression,baiocchi2014instrumental}. An instrumental variable is a variable that is associated with the treatment of interest but has no direct effect on the outcome and is independent of unmeasured confounders.  For example, a second treatment that is randomly assigned and encourages taking the treatment of interest but otherwise has no effect on the outcome would be an instrumental variable \citep{holland1988causal}.  
Point identification of the ATE using IVs requires additional strong assumptions. 
Examples include a homogeneous treatment effect so that the local average treatment effect \citep{angrist1996identification} equals the ATE, the instrument is strong enough to drive the probability of being treated from zero to one \citep{angrist2010extrapolate} or general no-interaction assumptions derived by \citet{wang2018bounded}; all of these assumptions may be hard to satisfy in practice. Hence, partial identification methods are generally considered more reliable than methods that yield point identification because they emphasize what can be learned without relying on potentially untenable assumptions. However, there are instances when even partially identifying the treatment effect by using a second treatment as an IV remains indefensible because the second treatment is associated with unmeasured confounders and is not a valid IV \citep{swanson2018partial}.


\subsection{Motivating Example: the Effect of Smoking On the Blood Level of Cadmium}

Cadmium, a highly toxic metal, poses serious health risks, including kidney damage and cardiovascular diseases. Smoking could be a significant source of cadmium exposure, as tobacco plants readily absorb cadmium from contaminated soil. Understanding the causal effect of smoking on blood cadmium levels is essential for designing public health interventions to reduce exposure. This study leverages data from the National Health and Nutrition Examination Survey (NHANES), a comprehensive dataset that collects detailed health, dietary, and laboratory data from a representative sample of the U.S. population.

Smoking cannot ethically be randomized and the central challenge in estimating the causal effect of smoking on blood cadmium levels arises from unmeasured confounding variables. Risky health behaviors, such as poor dietary choices or excessive alcohol consumption, could be associated with both smoking and cadmium levels, making it difficult to disentangle the specific contribution of smoking. These unmeasured factors introduce biases that can compromise traditional observational analyses.

While IV is a common method for addressing unmeasured confounding, finding a valid and strong IV in this context is highly challenging and often impractical. To address this issue, we turn to the differential comparison approach introduced by \citet{rosenbaum2006differential}. Individuals who have used hard drugs serve as a secondary group for comparison because such individuals often share a general propensity for risky or unhealthy behaviors, which could also influence smoking. By comparing daily smokers who have never used hard drugs with nonsmokers who have, we attempt to isolate the specific effects of smoking on cadmium levels, independent of a shared underlying inclination for risk-taking.

\subsection{Our Contribution: Partial Identification via Differential Effects}

Suppose the treatment of interest is treatment $A$ and there is also available a second treatment $B$.  The differential effect of treatment $A$ versus $B$ is the effect of applying $A$ in lieu of applying $B$. Under the Rasch model  \citep{rasch1980probabilistic, rasch1993probabilistic} for treatment assignment, \citet{rosenbaum2006differential} showed that the differential effect is not biased by generic biases that promote both treatments. A study of treatment $A$ vs. the control alone could be strongly biased by a generic bias, whereas the differential effect can be completely unbiased under certain model assumptions like a Rasch model \citep{rosenbaum2006differential}.  
\cite{rosenbaum2013using} advocates routine use of differential effects:  ``Differential comparisons [i.e., estimation of the differential effect] are easy to do, are quite intuitive, are supported in a specific way for a specific purpose by statistical theory, and are-at worst-a harmless and transparent supplement to standard comparisons of people who look similar. Yet, differential comparisons are rarely done.'' 

\cite{rosenbaum2006differential} shows that the differential effect point identifies the treatment effect under strong assumptions such as the Rasch model for treatment.  But if these strong assumptions do not hold, when can we learn something from differential effects? In this paper, we show that under certain semi-parametric assumptions, the ATE can be bounded by the probability limits of the differential effect estimate and the inverse probability weighting estimate and thus we can bracket the true effect.  We also provide a similar result for the CATE, and provide conditions under which the ATE and CATE can be point identified using the proposed method. We develop a flexible and easy-to-implement semi-parametric framework to estimate bounds for ATE and CATE and leverage a two-stage approach to conduct the inference for effects of interest under some mild conditions. Our approach offers several advantages. It is more practical than partial identification via IVs in many scenarios, particularly when identifying a valid and strong IV is challenging or infeasible, as our approach allows the second treatment to be correlated with unmeasured confounders. The assumptions underlying the method are straightforward to interpret and can be directly evaluated using the data. Additionally, the approach provides flexibility by accommodating both parametric and non-parametric models, allowing it to address a wide range of data structures and analytical needs. To examine the performance of our proposed method, we utilize the proposed framework in the setting of continuous outcomes with various factors in detail via a simulation study. Additionally, we apply the proposed methodology to examine the effect of smoking on the blood level of cadmium using data from the National Health and Nutrition Examination Survey (NHANES). 

\subsection{Related Work and Outline}

There is a lot of literature discussing bounds on the ATE and CATE using IV since \citet{manski1990nonparametric} pioneered partial identification of the ATE under the mean independent assumption of the instrument.   \citet{balke1997bounds} considered the setting of a randomized experiment
with noncompliance to bound the ATE on a binary outcome. \citet{shaikh2011partial} imposed threshold crossing models on both the treatment and
the outcome (with the latter paper focusing on a binary outcome). \citet{flores2013partial} established partial identification of local average treatment effects with an invalid instrument. Later, \citet{chen2016bounds} derived bounds on the ATE and the average treatment effect on the treated using an instrumental variable that does not satisfy the exclusion restriction assumption. \citet{yadlowsky2018bounds} developed a loss minimization strategy for obtaining bounds on the ATE and CATE when hidden bias has a bounded effect on the odds ratio of treatment selection without instruments. \citet{flores2018average} provided a comprehensive review with the frameworks for partial identification of the ATE of a time-fixed binary treatment on a binary outcome using a binary IV. However, none of the above consider the partial identification of the ATE and CATE for violations of the IV independent of unmeasured confounders assumption.

Although partial identification via instrumental variables (IVs) is a widely used approach, it typically relies on two key assumptions: (1) the instrument affects the outcome only through its potential impact on the treatment, and (2) the instrument is independent of unmeasured confounders. However, it is hard to find valid and strong IVs in many settings. While methods for partial identification using invalid IVs, such as those proposed by \citet{chen2016bounds}, have been developed, they often yield bounds that are too wide to provide meaningful insights. More restrictive assumptions are needed to achieve narrower bounds.  

The remainder of the paper is structured as follows. Section \ref{sec:settings} introduces the settings and preliminary assumptions of our framework. Section \ref{sec: treatment effects} establishes the bounds and asymptotic properties of proposed estimators for treatment effects. We examine the proposed framework through a simulation study in Section \ref{sec:simulation} and a case study in Section \ref{sec:case studies}. We conclude our paper in Section \ref{sec:discussion}. Additional theoretical discussions, proofs,  simulation results, and case study results can be found in the Supplementary material. Code is available on \emph{https://github.com/dosen4552/Differential\_Effect\_ATE\_CATE}.

\section{Settings and Preliminary Assumptions} \label{sec:settings}

We use the Neyman-Rubin potential outcome framework \citep{neyman1923application,rubin1974estimating}. The binary treatment of interest is $Z_1 \in \{0,1\}$ and there is a  second binary treatment assignment $Z_2 \in \{0,1\}$, a set of observed covariate $X \in \mathbb{R}^l $, $l \geq 1$, and potential outcomes $Y(1), Y(0) \in \mathbb{R}$ where $Y(1)$ is the potential outcome if taking level $1$ of $Z_1$ and $Y(0)$ is the potential outcome if taking level $0$ of $Z_1$. And the observed outcome $Y = Z_1 Y(1) + (1-Z_1) Y(0) = Y(Z_1)$ (consistency).  Our interest is to study the partial identification and inference of the \emph{average treatment effect} $\tau := \mathbb{E}[Y(1) - Y(0)]$
and \emph{conditional average treatment effect} $ \tau (x) := \mathbb{E}[Y(1) - Y(0)|X=x ]$ in the presence of unmeasured confounding based on $n$ independent and identically distributed observations $\bm W_i = \{Y_i = Y_i(Z_{i1}), Z_{i1}, Z_{i2}, \mathbf{X}_i \}_{i=1}^n$. Note that we are only interested in the causal effect of taking $Z_1$.  The no unmeasured confounding assumption is often not realistic in that there is usually an unmeasured confounding factor $U$ affecting both treatment assignment and outcome in practice. Hence, throughout this paper, we consider an unobserved confounder $U$ following the assumption below:
\begin{assumption}[Unconfoundedness] \label{assump: confoundness}
$\{Y(1), Y(0) \} \indep Z_1 |  X, Z_2, U$. 
\end{assumption}

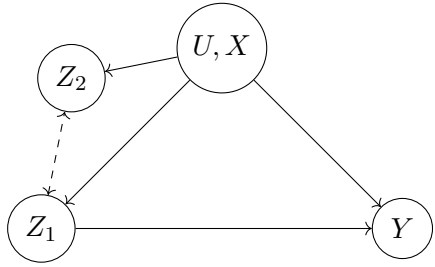
\begin{figure}
    \centering
    \begin{tikzpicture}[node distance=1cm, scale = 0.8]
    \tikzstyle{var} = [draw, circle, minimum width=0.8cm, minimum height=0.8cm]
    
    
        \node[var] (Z_1) at (0,0) {$Z_1$};  
        \node[var] (Y) at (6,0) {$Y$};  
        \node[var] (U) at (3,3) {$U, X$}; 
        \node[var] (Z_2) at (0.5,2.5) {$Z_2$};  
    
    \draw[->] (Z_1) -- (Y);
    \draw[->] (U) -- (Z_1);
    \draw[->] (U) -- (Z_2);
    \draw[->] (U) -- (Y);
    \draw[dashed, <->] (Z_1) -- (Z_2);
    
\end{tikzpicture}
\caption{A graphical representation of relationships among covariates $U, X$, first treatment assignment $Z_1$, second treatment assignment $Z_2$, and outcome $Y$ that satisfy Assumptions \ref{assump: confoundness} and \ref{assump: ER}. $Z_2$ is not a valid instrumental variable (IV) since it violates the IV independent of unmeasured confounders assumption.  }
\label{fig:DAG}
\end{figure}

We will also impose an assumption regarding the second treatment $Z_2$:

\begin{assumption}[Exclusion Restriction] \label{assump: ER}
    $\mathbb{E} [Y(z_1)|Z_2 = 1, X, U  ] =  \mathbb{E} [Y(z_1)|Z_2 = 0, X, U  ]$  for $z_1 \in \{0,1\}.$
\end{assumption}

Assumption \ref{assump: ER} says that if we consider the outcomes for $z_1$, the value of $Z_2$ does not affect the outcomes given observed covariate $X$ and the unobserved covariate $U$ {\textendash} a direct effect of $Z_2$ is ``excluded.'' Note that this assumption is similar in spirit to (and implied by) the standard exclusion restriction in IV analysis.  In the potential outcomes framework with IVs \citep{angrist1996identification}, the conventional exclusion restriction assumes $Y(z_1,1) = Y(z_1,0)$ for all $z_1 \in \{0,1\}$ for all subjects where $Y(z_1, z_2)$ denotes the potential outcome under treatment status $(z_1, z_2)$ \citep{angrist2009mostly}; this implies that $\mathbb{E}[Y(z_1)|Z_2 = 1, X, U] = \mathbb{E}[Y(z_1)|Z_2 = 0, X, U]$ for $z_1 \in \{0,1\}$. Additionally, $Z_2$ is not a valid instrumental variable because it violates the IV independent of unmeasured confounders assumption. The exclusion restriction assumption (assumption \ref{assump: ER}) can also be viewed as a no direct effect assumption. Note that we may relax assumption \ref{assump: ER} by imposing a direct constant effect of $Z_2$ on outcome. Detailed discussion can be found in the supplementary material (Appendix III). Additionally, we impose the positivity assumption for the propensity score \citep{hernan2006estimating}. 

\begin{assumption}[Positivity] \label{assump: positivity}
$ 0 < \mathbb{E} [Z_1|Z_2,X] = \mathbb{P}(Z_1 = 1 | Z_2, X) < 1$ and $ 0 < \mathbb{E} [Z_2|X] = \mathbb{P}(Z_2 = 1 | X) < 1$.  
    
\end{assumption}

\section{Identification of Treatment Effects}
\label{sec: treatment effects}

\subsection{Bounds on Treatment Effects}

We first define
\begin{equation} \label{def: bound tau(x)}
\begin{split}
        \mu_1(x) &= \mathbb{E}[Y|Z_1 = 1, Z_2 = 0, X=x] -  \mathbb{E}[Y|Z_1 = 0, Z_2 = 1,X=x], \\
        \mu_2(x) &= \mathbb{E}\left[\frac{Z_1 Y}{\mathbb{E} [Z_1|Z_2,X] } \Bigg| X=x\right] - \mathbb{E}\left[\frac{(1 -Z_1) Y}{\mathbb{E} [1 - Z_1|Z_2,X] } \Bigg| X=x\right],\\
        \tau^+(x) &= \max \{ \mu_1(x), \mu_2(x)  \}, \\
        \tau^-(x) &= \min \{ \mu_1(x), \mu_2(x)  \},
\end{split}
\end{equation}
where $\mu_1(x)$ is the conditional differential effect while $\mu_2(x)$ is the conditional inverse probability weighting given covariate $x$. To bound treatment effects, we make the following monotonicity assumptions on the (latent) outcome regression and propensity score.

\begin{assumption}
\label{assumption: outcome_mono}
$\mathbb{E} [Y \mid Z_1 = z_1, U = u, 
 X = x]$ is monotone in $u$ for $z_1 \in \{0,1\}$ and any fixed $x$.
\end{assumption}
\begin{assumption}
\label{assumption: treatment_mono}
$\mathbb{P}(Z_1 = 1 \mid Z_2 = z_2, U = u, X = x)$ is monotone in $u$ for $z_2 = 0,1$ and any fixed $x$.
\end{assumption}

These two assumptions require that the unmeasured confounder has a consistent directional effect on the average outcome across both treatment arms and on the probability of receiving treatment. In other words, they rule out certain sign interactions between $U$ and $X$ on the outcome $Y$ and treatment $Z_1$. Such assumptions have been invoked in previous work to assess or reduce bias from unmeasured confounders \citep{VanderWeele2008, Ogburn2012}. In particular, \cite{Ogburn2012} argue that such assumptions ``will likely hold in most applications in epidemiology.'' In the context of our example, where the unmeasured confounder represents risky behavior, these assumptions hold if, for both smokers and non-smokers, and across all levels of observed confounders $X$, a higher level of risky behavior increases the average level of cadmium in the blood. Similarly, for all levels of $X$ and all combinations of hard drug use, a higher level of risky behavior must also increase the probability of being a smoker.

In addition to monotonicity assumptions, we impose an assumption on the effect of $U$ on $Z_1$ and $Z_2$:

\begin{assumption}[$U$ influences $Z_2$ more than $Z_1$]
\label{assumption: z_2_influenced_more}
In the case where the direction of monotonicity of Assumption \ref{assumption: treatment_mono} is non-decreasing, $ \mathbb{P}(Z_1 = 1, Z_2 = 0 \mid U= u, X= x)$ is non-increasing in $u$, $\mathbb{P}(Z_1 = 0, Z_2 = 1 \mid U = u, X = x)$ is non-decreasing in $u$ for any fixed $x$.  In the case where the direction of monotonicity of Assumption \ref{assumption: treatment_mono} is non-increasing, $\mathbb{P}(Z_1 = 1, Z_2 = 0 \mid U = u, X = x)$ is non-decreasing in $u$, $\mathbb{P}(Z_1 = 0, Z_2 = 1 \mid U = u, X = x)$ is non-increasing in $u$ for any fixed $x$. 
\end{assumption}
The assumption essentially requires that $U$ has the same sign relationship with $Z_1$ as with $Z_2$ for all $x$, and when considering the mixed joint probabilities, the relationship between $u$ and $\mathbb{P}(Z_1 = z_1, Z_2 = z_2 \mid U = u, X = x)$ inherits the same monotonic relationship as $\mathbb{P}(Z_2 = z_2 \mid U = u, X = x)$. This would require that $U$ has a larger influence on $Z_2$ than on $Z_1$. In the context of the data application, Assumption \ref{assumption: z_2_influenced_more} would be satisfied if the probability of being a smoker and not a hard drug user decreases as risky behavior increases at all levels of $X$, and that the probability of being a non-smoker and a hard drug user increases as risky behavior increases at all levels of $X$. The plausibility of this assumption can be generally assessed by fitting a two-parameter logistic model from item response theory (IRT) \citep{hambleton2013item} where the item responses are the two treatments $Z_1$ and $Z_2$:

\begin{align*}
    \mathbb{P} (Z_1 = 1 |U = u, X = x) = \frac{\exp\{\alpha_1 u   + \beta_1 x \} }{1 + \exp\{ \alpha_1 u + \beta_1 x \} },
\end{align*}
and
\begin{align*}
        \mathbb{P} (Z_2 = 1 | U = u, X = x) = \frac{\exp\{\alpha_2 u  + \beta_2 x \} }{1 + \exp\{ \alpha_2 u  + \beta_2 x \} },
\end{align*}
 where $u$ is the latent variable, $\alpha_1$ and $\alpha_2$ are discrimination parameters while $\beta_1$ and $\beta_2$ are coefficients of covariate $x$. In our context, the latent variable $u$ is the unmeasured confounding, and $x$ is age; $\alpha_1, \alpha_2$  determine the rate at which the probability of taking the corresponding treatment $Z_1$ or $Z_2$ changes as a function of the level of unmeasured confounding. By the local independence assumption from IRT, that is, item responses are independent conditional on  the latent variable, we have
\begin{equation} \label{equ:rasch model}
    \begin{split}
           &\qquad \mathbb{P}(Z_1 = a, Z_2 = b | U = u, X = x) \\
    &= \frac{\exp \{a  (\alpha_1 u + \beta_1 x) + b  (\alpha_2 u + \beta_2 x) \} }{1 + \exp\{\alpha_1 u + \beta_1 x \} + \exp\{\alpha_2 u + \beta_2 x \} + \exp\{\alpha_1 u + \beta_1 x + \alpha_2 u + \beta_2 x \}}.
    \end{split}
\end{equation}

Hence, the first part of the assumption \ref{assumption: z_2_influenced_more} holds if $\alpha_2/\alpha_1 \geq \exp(\alpha_1 u + \beta_1 x) + \exp(\alpha_1 u + \beta_1 x + \alpha_2 u + \beta_2 x)$ and $\alpha_2/\alpha_1 \geq \exp(\alpha_2 u + \beta_2 x) + \exp(\alpha_1 u + \beta_1 x + \alpha_2 u + \beta_2 x)$ for all $x$ and $u$ while the second part of the assumption 6 holds if $\alpha_2/\alpha_1 \leq \exp(\alpha_1 u + \beta_1 x) + \exp(\alpha_1 u + \beta_1 x + \alpha_2 u + \beta_2 x)$ and $\alpha_2/\alpha_1 \leq \exp(\alpha_2 u + \beta_2 x) + \exp(\alpha_1 u + \beta_1 x + \alpha_2 u + \beta_2 x)$ for all $x$ and $u$. Detailed calculations can be found in the supplementary material (Appendix II). With the assumptions listed above, the main result for the identification of treatment effects can be stated.

\begin{theorem}[Identification of Treatment Effects] \label{thm: bound treatment effects}

Under assumption \ref{assump: confoundness}, \ref{assump: ER}, \ref{assump: positivity}, \ref{assumption: outcome_mono}, \ref{assumption: treatment_mono}, and \ref{assumption: z_2_influenced_more}, $\tau (x) \in [\tau^- (x), \tau^+ (x)]$ for any fixed $x \in \mathbb{R}^l$.
Additionally, denote $\mu_1 = \mathbb{E}_X [\mu_1(X)], \mu_2 = \mathbb{E}_X [\mu_2(X)], \tau^+ = \max\{\mu_1,\mu_2\}$ and $\tau^- = \min\{\mu_1,\mu_2\}$, under the same set of assumptions, $\tau \in [\tau^-, \tau^+]$. 
\end{theorem}

The proof of Theorem \ref{thm: bound treatment effects} can be found in the supplementary material (Appendix I).

\subsection{Semi-parametric Estimation of Bounds and Inference on Treatment Effects}

 To estimate the bounds, we propose the following semi-parametric estimators motivated by Nadaraya-Watson kernel regression \citep{nadaraya1964estimating,hardle1990applied,hardle1991smoothing}:

\begin{equation} \label{def: bound tau hat (x)}
\begin{split}
       \hat \mu_1 (x) & =  \frac{\frac{1}{n_{10}h_1^{l}} \sum_{i \in \mathcal{I}_1  } K_{1} (\frac{x - X_i}{h_1}) Y_i }{\frac{1}{n_{10}h_1^{l}} \sum_{i \in \mathcal{I}_1  } K_{1} (\frac{x - X_i}{h_1})} - \frac{\frac{1}{n_{01}h_2^{l}} \sum_{i \in \mathcal{I}_2  } K_{2} (\frac{x - X_i}{h_2}) Y_i }{\frac{1}{n_{01}h_2^{l}} \sum_{i \in \mathcal{I}_2  } K_{2} (\frac{x - X_i}{h_2})},  \\
       \hat \mu_2 (x) &=  \frac{\frac{1}{n h_3^{l}}  \sum_{i=1}^n  \left( \frac{Z_1Y_i}{\mathbb{\hat P}_n (Z_{i1} = 1|Z_{i2}, X_i) } - \frac{(1-Z_1)Y_i}{1 - \mathbb{\hat P}_n (Z_{i1} = 1|Z_{i2}, X_i) }  \right) K_{3} (\frac{x-X_i}{h_3})    }{\frac{1}{n h_3^{l}} \sum_{i=1}^n K_{3} (\frac{x-X_i}{h_3})  }, \\
      \hat \mu_1 &= \frac{1}{n} \sum_{i=1}^n \hat \mu_1(X_i) \quad \text{and } \quad \hat \mu_2 = \frac{1}{n} \sum_{i=1}^n \hat \mu_2(X_i)
\end{split}
\end{equation}
where $K_1(\cdot), K_2(\cdot), K_3(\cdot)$ are three kernel functions with bandwidths $h_1$,$h_2, h_3$ and order $s_1, s_2, s_3$, respectively. Note that the Nadaraya-Watson estimates $\hat \mu_1 (x), \hat \mu_2 (x)$ are consistent estimators of $\mu_1 (x), \mu_2(x)$ if $h_1 \rightarrow 0, nh_1 \rightarrow \infty, h_2 \rightarrow 0, nh_2 \rightarrow \infty, h_3 \rightarrow 0, nh_3 \rightarrow \infty$, and $\sup_{X \in \mathcal{X}, Z_2 \in \{0,1\}} |\mathbb{P}(Z_1 = 1| Z_2, X) - \hat{\mathbb{P}}_n (Z_1 = 1| Z_2, X)  | = O_p (n^{-1/2}) $ hold.

Next, we state results for the variances of $\hat \mu_1 (x)$ and $\hat \mu_2 (x)$. Let $m_{10}(x) = \mathbb{E}[Y|Z_1 = 1, Z_2 = 0, X=x], m_{01}(x) = \mathbb{E}[Y|Z_1 = 0, Z_2 = 1,X=x]$ with corresponding estimators 

$\hat m_{10}(x) = \Big( \frac{1}{n_{10}h_1^{l}} \sum_{i \in \mathcal{I}_1  } K_{1} (\frac{x - X_i}{h_1}) Y_i \Big) \Big/ \left( \frac{1}{n_{10}h_1^{l}} \sum_{i \in \mathcal{I}_1  } K_{1} (\frac{x - X_i}{h_1}) \right)$, 

$\hat m_{01}(x) = \Big( \frac{1}{n_{01}h_1^{l}} \sum_{i \in \mathcal{I}_2  } K_{2} (\frac{x - X_i}{h_1}) Y_i  \Big) \Big/ \Big( \frac{1}{n_{01}h_1^{l}} \sum_{i \in \mathcal{I}_2  } K_{2} (\frac{x - X_i}{h_1}) \Big)$. 

Additionally, denote $f_{10}(x)$ the density of $X$ at $x$ conditional on $Z_1 = 1$ and $Z_2 = 0$, $f_{01}(x)$ the density of $X$ at $x$ conditional on $Z_1 = 0$ and $Z_2 = 1$, $f(x)$ the density of $X$ at $x$, $\sigma_1^2(x) = Var (\hat{\mu}_1 (x) )$, and $\sigma_2^2(x) = Var (\hat{\mu}_2 (x) )$. Under some mild conditions, 
\begin{align*}
    \sigma_1^2(x) = \frac{\sigma_{10}^2 (x) \|K_1 \|_2^2}{f_{10}(x)} + \frac{\sigma_{01}^2 (x) \|K_2 \|_2^2}{f_{01}(x)}, \quad
     \sigma_2^2(x) = \frac{\sigma_{\phi}^2 (x) \|K_3 \|_2^2}{f(x)}  
\end{align*}
where 
\begin{align*}
    &\|K \|_2^2 = \int K^2(u)du, \quad \sigma_\phi^2 (x) = \mathbb{E} \Big[ \left( \phi (X, Y, Z_1, Z_2) - \mu_2 (x) \right)^2 | X= x  ],  \\
    &\phi (x,y,z_1, z_2) = z_1 y / p(x,z_2) - (1-z_1)y/\left(1 - p(x,z_2)\right), \\
    &p(x, z_2) = \mathbb{P}(Z_1 = 1| Z_2 = z_2, X = x), \\
    &\sigma^2_{10}(x) = \mathbb{E} \Big[ \left(\hat m_{10} (X) - m_{10} (X) \right)^2 |X = x  \Big], \sigma^2_{01}(x) = \mathbb{E} \Big[ \left(\hat m_{01} (X) - m_{01} (X) \right)^2 |X = x  \Big].  
\end{align*}
Detailed justification can be found in the supplementary materials (Appendix V). To estimate the variances, we only need estimates of $\sigma_{10}^2(x), \sigma_{01}^2(x),\sigma_\phi^2(x), f_{10}(x), f_{01}(x)$ and $f(x)$. The estimates of $f_{10}(x), f_{01}(x), f(x)$ are the ordinary kernel density estimators $\hat f_{10}(x), \hat f_{01}(x), \hat f(x)$, and natural estimates of $\sigma_{10}^2(x), \sigma_{01}^2(x),$ and $\sigma_\phi^2(x)$ can defined as

\begin{align*}
    \hat \sigma_{10}^2 (x) &= n_{10}^{-1} \sum_{i \in \mathcal{I}_1} \frac{h_1^{-1}K_1(\frac{x-X_i}{h_1})  }{\hat f_{10}(x)} \left( Y_i - \frac{\frac{1}{n_{10}h_1^{l}} \sum_{j \in \mathcal{I}_1  } K_{1} (\frac{x - X_j}{h_1}) Y_j }{\frac{1}{n_{10}h_1^{l}} \sum_{j \in \mathcal{I}_1  } K_{1} (\frac{x - X_j}{h_1})} \right)^2 \\ 
    \hat \sigma_{01}^2 (x) & =  n_{01}^{-1} \sum_{i \in \mathcal{I}_2} \frac{h_2^{-1}K_2(\frac{x-X_i}{h_2})  }{\hat f_{01}(x)} \left( Y_i - \frac{\frac{1}{n_{01}h_2^{l}} \sum_{j \in \mathcal{I}_2  } K_{2} (\frac{x - X_j}{h_2}) Y_j }{\frac{1}{n_{01}h_2^{l}} \sum_{j \in \mathcal{I}_2  } K_{2} (\frac{x - X_j}{h_2})} \right)^2\\
    \hat \sigma_\phi^2 (x) &= n^{-2} \sum_{i=1}^n  \left( \hat \phi (X_i, Y_i, Z_{i1}, Z_{i2}) - \hat \mu_2 (x) \right)^2 \\
     \hat \phi (x, y, z_1, z_2) &=  \frac{z_1 y}{\mathbb{\hat P}_n (Z_{1} = 1|Z_{2} = z_2, X = x) } - \frac{(1-z_1)y}{1 - \mathbb{\hat P}_n (Z_{1} = 1|Z_{2} = z_2, X = x) },   
\end{align*}

where $\mathbb{\hat P}_n (Z_{i1} = 1|Z_{i2},X_i) $ is the estimated propensity by a parametric model (e.g. a logit or probit model estimated by maximum likelihood). To address the bias in nonparametric estimators, we employ undersmoothing by using kernel regression and selecting a bandwidth smaller than the optimal value that minimizes the asymptotic mean integrated squared error.

To construct a confidence interval for $\tau (x)$, we implement the Generalized Moment Selection (GMS) proposed by \citet{andrews2013inference,andrews2017commands} for testing conditional moment inequalities. We abuse the notation a bit here. Let $\hat \mu_1^{\bm W}(x) = \hat \mu_1(x)$ and $\hat \mu_2^{\bm W}(x) = \hat \mu_2(x)$ where $\bm W = (Y, Z_1, Z_2,\bm X)$. This is because we need $\bm W_i = (Y_i, Z_{1i}, Z_{2i}, \bm X_i)$ to obtain $\hat \mu_1(x)$ and $\hat \mu_2(x)$. Without loss of generality, assuming $\mu_1 (x) \leq \tau (x) \leq \mu_2 (x)$,  we will consider tests of the null hypothesis
\begin{align*}
    H_0: \mathbb{E} \left( \tau (x) - \hat \mu_1^{\bm W} (x)   \mid \bm X =  x \right) \geq 0, \mathbb{E}\left( \hat \mu_2^{\bm W} (x) - \tau (x) \mid \bm X =  x \right) \geq 0
\end{align*}
for $\tau (x) \in \mathbb{R}$ at fixed $x \in \mathbb{R}^l$ that control the probability of Type I error at level $\alpha$. GMS can be separated into two steps. In the first step, we compute the test statistic. In the second step, we compute the GMS critical value and construct the corresponding confidence interval. We list the algorithm 1 tailored for our problem below. 

\textbf{Algorithm 1. (Confidence Interval for CATE)}
\begin{itemize}
   \item[] Step 1: Compute Test Statistics. 
   \begin{itemize}
       \item[(a).] Transform each regressior to lie in $[0,1]$: Let $\bm X_i^\dagger$ denote the untransformed regressor vector, and 
       \begin{align*}
           \hat{\Sigma}_{\bm X,n} = n^{-1} \sum_{i=1}^n (\bm X_i^\dagger - \bar{\bm X}_n^\dagger   ) (\bm X_i^\dagger - \bar{\bm X}_n^\dagger   )^\top
       \end{align*}
       where $ \bar{\bm X}_n^\dagger = n^{-1} \sum_{i=1}^n \bm X_i^\dagger $. Let $\bm X_i = \Phi( \hat{\Sigma}_{\bm X,n}^{-1/2} (\bm X_i^\dagger - \bar{\bm X}_n^\dagger   )   )$ where $\Phi(\bm X) = (\Phi(x_1), \cdots, \Phi(x_l) )^\top$ for $x = (x_1, \cdots, x_l)$, and $\Phi(x_j)$ is the standard normal distribution for $x_j \in \mathbb{R}$.  
       \item[(b).] Specify the function $g$. 
       \begin{align*}
           g_{\bm a,r}(x) = 1 \{ x \in C_{\bm a,r}    \} \mathbf{1}_l
       \end{align*}
       where $C_{\bm a,r} = \times_{u=1}^l \left( (a_u - 1)/(2r), a_u/(2r)   \right]$, $\bm a = (a_1, \cdots,a_l), a_u \in \{1,2,\cdots,2r \}, \mathbf{1}_l = ( \underbrace{1,\cdots,1}_l), r = r_0, r_0 + 1, \cdots$. 
       \item[(c).] Specify the weight function. We take it to be uniform on $\bm a \in \{1, \cdots,2r \}^l$ given $r$, combined with $w(r) = (r^2 + 100)^{-1}$ for $r = 1, \cdots, r_{1,n}$. We choose $r_{1,n}$ be $4$ in the case study and simulations.
       \item[(d).] Compute the test statistic, which is defined by
     \begin{equation} \label{test stats}
       \begin{split}
                    \bar T_{n, r_{1,n}} (\tau(x)  ) = &\sum_{r=1}^{r_{1,n}} (r^2+100)^{-1} \\
                    & \times \sum_{ a\in \{1,\cdots,2r \}^l   } (2r)^{-l} S\left(n^{1/2} \bar{m}_n (\tau(x), g_{\bm a, r}   ), \bar{\Sigma}_n (\tau(x), g_{\bm a, r}  )   \right)
       \end{split}
     \end{equation}
     where 
     \begin{align*}
         S(\bm m, \Sigma) = [m_1/\sigma_1^2]_-^2 + [m_2/\sigma_2^2]_-^2
     \end{align*}
     where $m_j$ is the $j$-th element of the vector $\bm m$, $\sigma_j^2$ is the $j$-th diagonal element of the matrix $\Sigma$, and $[x]_- = -x$ if $x<0$ and $[x]_- = 0$ if $x \geq 0$,
     \begin{align*}
         &\bar m_n( \tau (x), g_{\bm a, r} ) = n^{-1} \sum_{i=1}^n m(\bm W_i, \tau (x), g_{\bm a, r} ), \\
         &m(\bm W_i, \tau (x), g_{\bm a, r} ) = \begin{bmatrix}
            \left(\tau (x) - \hat \mu_1^{\bm W} (x) \right) g_{\bm a, r} (x)   \\
           \left(\hat \mu_2^{\bm W} (x) - \tau (x) \right) g_{\bm a, r} (x)  
         \end{bmatrix}, \\
         &\bar{\Sigma}_n (\tau(x), g_{\bm a, r}  ) = \hat{\Sigma}_n (\tau(x), g_{\bm a, r} ) + \epsilon \text{Diag}(\hat{\Sigma}_n (\tau(x), g_{\bm a, r} ) ), \\
         & \hat{\Sigma}_n (\tau(x), g_{\bm a, r} ) = n^{-1} \sum_{i=1}^n \left( m(\bm W_i, \tau (x), g_{\bm a, r} ) - \bar m_n( \tau (x), g_{\bm a, r} )   \right)  \left( m(\bm W_i, \tau (x), g_{\bm a, r} ) - \bar m_n( \tau (x), g_{\bm a, r} ) \right)^\top,
     \end{align*}
     with $\epsilon = 0.05$. 
   \end{itemize}
   \item[] Step 2: Compute the GMS critical value based on the asymptotic distribution.
   \begin{itemize}
       \item[(a).] Compute $\psi_n (\tau(x), g_{\bm a, r}   ) $ where 
       \begin{align*}
           &\psi_n (\tau(x), g_{\bm a, r}   ) = \left( \psi_{n,1} (\tau(x), g_{\bm a, r}), \psi_{n,2} (\tau(x), g_{\bm a, r}    )     \right), \\
           &\psi_{n,j} (\tau(x), g_{\bm a, r} ) = B_n 1 \{ \xi_{n,j} (\tau(x), g_{\bm a, r}) > 1 \}, j = 1, 2, \\
           &\xi_{n} (\tau(x), g_{\bm a, r}) = \kappa_n^{-1} n^{1/2} \bar{D}_n^{-1/2} (\tau (x), g_{\bm a, r} ) \bar m_n( \tau (x), g_{\bm a, r} ), \\
           &\bar{D}_n^{-1/2}  (\tau (x), g_{\bm a, r} ) = \text{Diag}(\bar{\Sigma}_n (\tau(x), g_{\bm a, r} )  ),
       \end{align*}
       where we take $\kappa_n = (0.3 \ln(n))^{1/2}, B_n = \left(0.4 \ln(n) / \ln (\ln (n))  \right)^{1/2}$, and $\xi_{n,j} (\tau(x), g_{\bm a, r})$ is the $j$-th component of $\xi_{n} (\tau(x), g_{\bm a, r})$.
       \item[(b).] Simulate a $(2 N_g) \times N_{rep} $ matrix $Z$ of standard normal random variable where $N_g = \sum_{r=1}^{r_{1,n}}(2r)^l$ and we set $N_{rep} = 500$. 
       \item[(c).] $(2 N_g) \times (2 N_g)$ covariance $\hat{h}_{2,n,mat}(\tau (x))$ whose elements are 
       \begin{align*}
           \hat{h}_{2,n} (\tau(x), g_{\bm a, r}, g_{\bm a, r}^*  ) &= \hat{D}_n^{-1/2}(\tau(x)) \hat{\Sigma}_n(\tau(x), g_{\bm a,r}, g_{\bm a,r}^*) \hat{D}_n^{-1/2} (\tau(x)) \\
           \hat{\Sigma}_n(\tau(x), g_{\bm a,r}, g_{\bm a,r}^*) &=  n^{-1} \sum_{i=1}^n \left( m(\bm W_i, \tau (x), g_{\bm a, r} ) - \bar m_n( \tau (x), g_{\bm a, r} )   \right) \\
           &\times \left( m(\bm W_i, \tau (x), g_{\bm a, r}^* ) - \bar m_n( \tau (x), g_{\bm a, r}^* ) \right)^\top \\
           \hat{D}_n^{-1/2}(\tau(x)) &= \text{Diag}(\hat{\Sigma}_n (\tau(x), \mathbf{1}_2, \mathbf{1}_2 ) )
       \end{align*}
       \item[(d).] Compute the $(2 N_g) \times N_{rep}$ matrix $\hat{\nu}_n (\tau (x) ) = \hat{h}_{2,n,mat}^{1/2}(\tau (x)) Z$. Let $\hat{\nu}_{n,j} (\tau (x) )$ denote the $k$-dimensional subvector of the matrix $\hat{\nu}_n (\tau (x) )$ that corresponds to the $k$ rows indexed by $g_{\bm a, r}$ and column $j$ for $j = 1, \cdots, N_{rep}$. 
       \item[(e).] For $j = 1, \cdots, N_{rep}$, compute the test statistic $\bar T_{n,r_{1,n},j}(\tau (x))$ just as $\bar T_{n,r_{1,n}}(\tau (x))$ is computed in step 1 (d) but with $n^{1/2} \bar m_n (\tau(x), g_{\bm a, r} )$ replaced by $\hat{\nu}_{n,j} (\tau (x) ) + \psi_n (\tau(x), g_{\bm a, r}   )$. 
       \item[(f).]  Take the critical value to be $1 - \alpha + \eta$ sample quantiles of $\{T_{n.r_{1,n},j} (\tau (x)): j = 1, \cdots, N_{rep} \}$ (say this is $T_{n.r_{1,n}}^* (\tau (x))$), and we set $\eta = 10^{-6}$. The $(1 - \alpha) \times 100\%$ confidence set is 
   \begin{align*}
       c_n(x) = \{\tau(x) \in \mathbb{R}:  \bar T_{n, r_{1,n}} (\tau(x)  ) < T_{n.r_{1,n}}^* (\tau (x))   \}
   \end{align*}
   where $x \in \mathbb{R}^l$. 
   \end{itemize}
   
\end{itemize}

We then have a corollary which demonstrates the correct coverage of the constructed confidence interval for CATE. 

\begin{corollary}[Coverage of the CATE]
    $c_n (x)$ obtained from step 2(f) in Algorithm 1 satisfies 
   \begin{align*}
       \lim_{n \rightarrow \infty} \inf \inf_{(\tau(x), F) \in \mathcal{F}, \hat h_{2,n,mat}(\tau(x)) \in \mathcal{H}_{2,cpt} } \inf_{\tau \in \mathbb{R}} \mathbb{P}_F \left(\tau(x) \in c_n (x ) \right) \geq 1 - \alpha 
   \end{align*}
   for any $x \in \mathbb{R}^l$ where $\mathcal{H}_{2,cpt}$ is compact sets of covariances kernels and $\hat h_{2,n,mat}(\tau(x))$ is defined in step 2(c) in Algorithm 1.  
\end{corollary}
This corollary shows that GMS tailed to our problem has correct uniform asymptotic size over compact sets of covariance kernels, and it is a direct implication of the Theorem 2 in \citet{andrews2013inference}. As for the inference for the ATE, detailed discussion can be found in the supplementary material (Appendix IV).

\section{Simulation} \label{sec:simulation}

\subsection{Goal and Design}

The goal of this simulation is to investigate the coverage probability of our proposed estimators for bounds of the ATE and CATE in the presence of unmeasured confounders under different data generating processes. The coverage probability is the probability that the  true ATE or CATE lies within the estimated bounds or corresponding confidence intervals. Consider the following data generating process. We observe $n$ independent and identically distributed $\{Y_i, Z_{1i}, Z_{2i}, \mathbf{X}_i \}$ where  $Z_{i2} \sim Ber(1/2), C_i \sim \text{Multinomial}\left((0,1,2), (p,(1-p)/2,(1-p)/2 ) \right)$,

\begin{align*}
      Z_{i1} &=  \begin{cases} 
      Z_{i2} & C_i = 0 \\
      1 & C_i = 1 \\
      0 & C_i = 2 
   \end{cases} \\
    U_i &= \kappa  + \delta Z_{i2} +  \nu_i
\end{align*}
with two different outcome models, namely, the homogeneous effect model: $    Y_i = \theta Z_{i1} + \mathbf{\gamma}^{\top} \mathbf{X}_i  + U_i + \epsilon_i$, and the heterogeneous effect model: $    Y_i = \left(\alpha \sin( X_{i1}/2) + \beta \right) Z_{i1} + \mathbf{\gamma}^{\top} \mathbf{X}_i  + U_i + \epsilon_i$,
$\nu_i \sim \mathcal{N}(0, \sigma^2), \epsilon_i \sim \mathcal{N} (0,  \omega^2)$, $\mathbf{X}_i | C_i = 0 \sim \mathcal{N} (\mathbf{0}_d, \mathbf{I}_{d \times d} ) $,  $\mathbf{X}_i | C_i = 1 \sim \mathcal{N} (\mathbf{0.25}_d, \mathbf{I}_{d \times d} ) $, and $\mathbf{X}_i | C_i = 2 \sim \mathcal{N} (\mathbf{0.50}_d, \mathbf{I}_{d \times d} ) $. We summarize the factors of the simulation as follows. 

\begin{itemize}
    \item[\textbf{Factor 1}] Choice of $\delta$: 1.) $\delta = 0$, and $\delta = 1$. This aims to examine the coverage probabilities when $Z_2$ is not correlated with the unmeasured confounder $U$ ($\delta =0$) and is correlated with the unmeasured confounder $U$ ($\delta =1$);
    \item[\textbf{Factor 2}] Number of observations: 1.) $n=1000$,  2.) $n=2000$,  and 3.) $n=5000$;
    \item[\textbf{Factor 3}] Dimension of observed covariates: 1.) $d = 5$,  2.) $d = 10$, and 3.) $d = 20$;
    \item[\textbf{Factor 4}] Outcome models: 1.) homogeneous effect model, and 2.) heterogeneous effect model;
    \item[\textbf{Factor 5}] Proportion of observations with $C_i = 0$: 1.) $p = 0.6$, 2.) $p = 0.7$, and 3.) $p = 0.8$. This aims to examine the coverage probabilities under different correlations between $Z_{i1}$ and $Z_{i2}$.
\end{itemize}

The rest of the parameters are as follows. $\kappa  =  \sigma^2 = \omega^2 = 1, \mathbf{\gamma} = \mathbf{1}_d, \theta = 3, \alpha = 2$, and $\beta = 2.5$. We will use the logistic model to estimate the propensity, the Gaussian kernel to estimate the CATE conditioned on the first component of $\mathbf{X}$ (i.e. $X_{1}$), and Asymptotic Mean Integrated Squared Error (AMISE) to obtain the optimal bandwidth. As suggested by \citet{abrevaya2015estimating}, to reduce bias results, for fully non-parametric estimator, we undersmooth the conditional outcome (i.e., $h_1, h_2$), and for the semi-parametric estimator, we undersmooth with $h_3$. If we use fully non-parametric estimator for CATE (which is not discussed in this paper), we should avoid undersmoothing for the propensity score. We set $x = X_{1} = 1$ for examining the coverage probability of the CATE. We will run $500$ simulations under different combinations of factors and report the estimated coverage probability. And we set the number of bootstrap samplers to be $1000$. 

\subsection{Results}

\begin{table}[]
\centering  
\begin{tabular}{  cccccccc  p{5cm} }
\hline
 &  Outcome model & $p = 0.7$ & $p = 0.8$ & $p = 0.9$ & $p = 0.7$ & $p = 0.8$ & $p = 0.9$ \\ 
\hline
   \multicolumn{7}{c}{\qquad \qquad \qquad \qquad \qquad \qquad \qquad $n = 1000,\delta = 0$ \quad \quad  \qquad \qquad  $n = 1000,\delta = 1$}   \\ 
  \multirow{2}{4em}{$d = 5$} & \emph{Homogeneous}  & 0.996 & 1.000 & 1.000 & 0.998 & 1.000 & 1.000 \\ 
   & \emph{Heterogeneous}  & 0.992 & 0.982 & 1.000 & 1.000 & 1.000 & 1.000 \\ 
  \multirow{2}{4em}{$d = 10$} & \emph{Homogeneous}  & 1.000 & 1.000 & 1.000 & 1.000 & 1.000 & 1.000 \\ 
   & \emph{Heterogeneous}  & 1.000 & 1.000 & 0.998 & 1.000 & 1.000 & 0.996 \\ 
 \multirow{2}{4em}{$d = 20$}  & \emph{Homogeneous}  & 1.000 & 1.000 & 1.000 & 1.000 & 1.000 & 1.000  \\ 
   & \emph{Heterogeneous} & 1.000 & 1.000 & 1.000 & 1.000 & 1.000 & 1.000  \\ 
   \hline
   \multicolumn{7}{c}{\qquad \qquad \qquad \qquad \qquad \qquad \qquad $n = 2000,\delta = 0$ \quad \quad  \qquad \qquad  $n = 2000,\delta = 1$}   \\ 
   
\multirow{2}{4em}{$d = 5$} & \emph{Homogeneous} & 0.990 & 0.990 & 0.996 & 1.000 & 1.000 & 1.000\\ 
   & \emph{Heterogeneous} & 0.986 & 0.986 & 0.990 & 1.000 & 1.000 & 1.000 \\ 
  \multirow{2}{4em}{$d = 10$} & \emph{Homogeneous} & 0.998 & 0.998 & 1.000 & 1.000 & 0.998 & 1.000 \\ 
   & \emph{Heterogeneous} & 1.000 & 0.998 & 1.000 & 1.000 & 1.000 & 1.000 \\ 
  \multirow{2}{4em}{$d = 20$} & \emph{Homogeneous} & 1.000 & 1.000 & 1.000 & 1.000 & 1.000 & 1.000 \\ 
   & \emph{Heterogeneous} & 1.000 & 1.000 & 1.000 & 1.000 & 1.000 & 1.000 \\
   \hline
   \multicolumn{7}{c}{\qquad \qquad \qquad \qquad \qquad \qquad \qquad $n = 5000,\delta = 0$ \quad \quad  \qquad \qquad  $n = 5000,\delta = 1$}   \\ 
  \multirow{2}{4em}{$d = 5$}  & \emph{Homogeneous}  & 0.994 & 0.994 & 0.990 & 1.000 & 1.000 & 0.998 \\ 
    & \emph{Heterogeneous}  & 0.994 & 0.992 & 0.978 & 1.000 & 1.000 & 1.000 \\ 
  \multirow{2}{4em}{$d = 10$}  & \emph{Homogeneous}  & 1.000 & 0.998 & 0.996 & 1.000 & 1.000 & 1.000  \\ 
    & \emph{Heterogeneous} & 1.000 & 0.998 & 0.998 & 1.000 & 1.000 & 1.000\\ 
  \multirow{2}{4em}{$d = 20$}  & \emph{Homogeneous}& 1.000 & 1.000 & 1.000 & 1.000 & 1.000 & 1.000 \\ 
    & \emph{Heterogeneous} & 1.000 & 1.000 & 1.000 & 1.000 & 1.000 & 1.000 \\ 
    \hline

\end{tabular}
\caption{Reported coverage probability of the $95 \%$ confidence interval of bounds for the ATE across different combinations of factors.}
\label{tab: CI.ATE}
\end{table}

\begin{table}[] 
\centering
\begin{tabular}{  cccccccc  p{5cm} }
\hline
 &  Outcome model & $p = 0.7$ & $p = 0.8$ & $p = 0.9$ & $p = 0.7$ & $p = 0.8$ & $p = 0.9$ \\ 
\hline
   \multicolumn{7}{c}{\qquad \qquad \qquad \qquad \qquad \qquad \qquad $n = 1000,\delta = 0$ \quad \quad  \qquad \qquad  $n = 1000,\delta = 1$}   \\ 
  \multirow{2}{4em}{$d = 5$} & \emph{Homogeneous}  & 1.000 & 0.996 & 0.960 & 0.994 & 0.996 & 0.970 \\
   & \emph{Heterogeneous}  & 1.000 & 1.000 & 0.962 & 1.000 & 0.996 & 0.966 \\ 
  \multirow{2}{4em}{$d = 10$} &\emph{Homogeneous}  & 1.000 & 1.000 & 0.970 & 0.998 & 0.994 & 0.980 \\ 
   & \emph{Heterogeneous}   & 1.000 & 0.998 & 0.980 & 0.998 & 0.998 & 0.976 \\ 
 \multirow{2}{4em}{$d = 20$}  & \emph{Homogeneous}  & 1.000 & 0.998 & 0.980 & 1.000 & 0.996 & 0.988 \\ 
   & \emph{Heterogeneous} & 1.000 & 0.992 & 0.986 & 1.000 & 0.998 & 0.972 \\ 
   \hline
   \multicolumn{7}{c}{\qquad \qquad \qquad \qquad \qquad \qquad \qquad $n = 2000,\delta = 0$ \quad \quad  \qquad \qquad  $n = 2000,\delta = 1$}   \\ 
\multirow{2}{4em}{$d = 5$} & \emph{Homogeneous} & 1.000 & 1.000 & 0.998 & 0.996 & 1.000 & 0.984 \\ 
   & \emph{Heterogeneous} & 1.000 & 1.000 & 0.988 & 0.998 & 0.996 & 0.982 \\ 
  \multirow{2}{4em}{$d = 10$} & \emph{Homogeneous} & 1.000 & 0.998 & 0.992 & 1.000 & 1.000 & 0.994 \\ 
   & \emph{Heterogeneous} & 1.000 & 1.000 & 0.998 & 1.000 & 1.000 & 0.992 \\ 
  \multirow{2}{4em}{$d = 20$} & \emph{Homogeneous} & 1.000 & 1.000 & 0.992 & 1.000 & 1.000 & 1.000 \\ 
   & \emph{Heterogeneous} & 1.000 & 1.000 & 0.998 & 1.000 & 1.000 & 0.996 \\ 
   \hline
   \multicolumn{7}{c}{\qquad \qquad \qquad \qquad \qquad \qquad \qquad $n = 5000,\delta = 0$ \quad \quad  \qquad \qquad  $n = 5000,\delta = 1$}   \\ 
  \multirow{2}{4em}{$d = 5$}  & \emph{Homogeneous} & 1.000 & 1.000 & 1.000 & 1.000 & 1.000 & 1.000 \\ 
    & \emph{Heterogeneous}  & 1.000 & 1.000 & 1.000 & 1.000 & 1.000 & 0.998 \\ 
  \multirow{2}{4em}{$d = 10$}  & \emph{Homogeneous}  & 1.000 & 1.000 & 1.000 & 1.000 & 1.000 & 0.998 \\ 
    & \emph{Heterogeneous} & 1.000 & 1.000 & 1.000 & 1.000 & 1.000 & 1.000 \\ 
  \multirow{2}{4em}{$d = 20$}  & \emph{Homogeneous}& 1.000 & 1.000 & 1.000 & 1.000 & 1.000 & 1.000 \\ 
    & \emph{Heterogeneous} & 1.000 & 1.000 & 1.000 & 1.000 & 1.000 & 1.000 \\ 
    \hline
\end{tabular}
\caption{Reported coverage probability of the $95 \%$ confidence interval of bounds for the CATE at $X_{1} = 1$ across different combinations of factors.}
\label{tab: CI.CATE}
\end{table}

Table \ref{tab: CI.ATE} and Table \ref{tab: CI.CATE} summarize the coverage probability of the $95 \%$ confidence interval of bounds for both the ATE and CATE under different data generating processes. Additional simulation results of the coverage probability and relative length of estimated bounds for both the ATE and CATE can be found in Table \ref{tab: ATE}, Table \ref{tab: CATE}, Table \ref{tab: relative length CI ATE}, and Table \ref{tab: relative length CI CATE} in the Supplementary material (Appendix VI).

We identify several consistent patterns from the simulation results. First, we find that whether $Z_2$ is correlated with the unmeasured confounder $U$ or not, our proposed methodology does well in the sense that $95 \%$ confidence intervals of bounds for both the ATE and CATE have a coverage probability close to 1 for any combinations of factors. Note that since we are estimating bounds and forming a confidence interval for the true effects using the two-stage approach, if the true causal effect lies within the interior of the probability limits of the bounds, then we would expect coverage probabilities greater than $95\%$ for confidence intervals for the true effect \citep{jiang2018using}. Furthermore, for the coverage probability of estimated bounds, the proposed methodology does substantially better asymptotically in the sense that the coverage probability gets closer to 1 as the number of observations increases. 

Note that we can adapt any kernels including Gaussian, Cosine, Sigmoid, etc as well as different kernel estimators such as Priestley–Chao \citep{priestley1972non}, and Gasser–Müller \citep{gasser1979kernel} in our framework which makes our approach flexible for implementing different model classes. In addition, parametric models for both conditional outcomes and treatment assignment could also be implemented.

\section{Case Study: the Effect of Smoking On the Blood Level of Cadmium } \label{sec:case studies}

The National Health and Nutrition Examination Survey is a set of surveys originally created to assess the health and nutritional status of adults and children in the United States. The NHANES includes demographic, socioeconomic, dietary, and health-related questions. 

Following \citet{rosenbaum2013using}, we use NHANES to examine the effect of smoking ($Z_1$) on blood levels of cadmium. Cadmium is a toxic metal with various detrimental health effects \citep{jarup2009current} so increases in cadmium would be concerning. We consider whether the person past used hard drugs,  in particular, the answer to the NHANES question ``Have you ever used cocaine, crack cocaine, heroin, or methamphetamine?" ($Z_2$) as the second treatment.  We exclude people who are currently using hard drugs from the study so $Z_2$ indicates past hard drug use but not current hard drug use. Compared to \citet{rosenbaum2013using} who used data from the 2005-2006 NHANES, we used updated data from the 2017-2018 NHANES.  A current smoker is defined as smoking at least 10 cigarettes every day for the past 30 days, while a non-current smoker is defined as not smoking cigarettes in the prior 30 days. We control for the following covariates: age, gender, ethnicity, education level, and the ratio of family income to poverty. We exclude individuals with missing data. Detailed information on the distribution of the covariates in smokers versus non-smokers can be found in Table \ref{tab:demographics NHANES} in Appendix VI in the Supplementary materials. We consider estimating bounds on both the ATE and CATE conditional on different genders and education levels.





\begin{table}[]
\centering
\begin{tabular}{  cccc  p{5cm} }
\hline
 & \emph{Covariate}   & \emph{Estimated bounds}   & \emph{95\% C.I bounds}    \\ 
\hline
\multirow{1}{4em}{\bf \emph{ATE}} &  &  $[10.24, 10.70]$   & $[8.13, 12.57]$ \\
\hline
   \multirow{4}{4em}{\bf \emph{CATE}}  & \emph{Gender(male)}    & \centering $[9.09,11.37]$ & $[8.04,12.00]$   \\ 
   & \emph{Gender(female)}  &  $[10.16, 11.25]$ & $[6.47,12.65]$ \\
      & \emph{Education(college)} & $[9.21, 10.45]$ & $[7.57, 11.32]$   \\ 
   & \emph{Education(high school)} & $[11.20,12.53]$ & $[9.21,13.53]$  \\
   \hline
\end{tabular}
\caption{Bounds of average and heterogeneous effects under exposure to smoking in the blood level of cadmium (nmol/L) by NHANES 2017-2018. We set the number of bootstrap samplers to be $1000$.}
\label{tab:NHANES results}
\end{table}

Since both smoking status and past usage of hard drugs could be correlated with an unmeasured confounder such as other unhealthy habits, past usage of hard drugs does not seem like a plausible IV. Instead, we consider the possibility of using past usage of hard drugs as the second treatment in a differential effects analysis. It is plausible that past usage of hard drugs satisfies the exclusion restriction because having tried drugs in the past but not continuing to use them into the present is unlikely to affect the present level of cadmium. To examine the plausibility of past usage of hard drugs satisfying the monotonicity assumption, we consider a two-parameter logistic model for health risky behaviors. Following the study from \citet{stimpson2007neighborhood}, we consider health risky behaviors of unhealthy fat intake and alcohol abuse in addition to smoking and past usage of hard drugs. Our assumptions are plausible because 1.) higher level of risky behavior increases the average level of cadimum and the probability of being a smoker, and 2.) the probability of being a non-smoker and a hard drug user increases as risky behavior increases at all level of $X$ (see Appendix II for details). Since the assumptions for our proposed method seem plausible,  we implement our proposed method for the identification of the ATE and CATE. Estimates of the bounds and associated $95\%$ confidence intervals for the ATE and CATE of the effect of smoking status on the blood level of lead with past usage of hard  drugs are presented in Table \ref{tab:NHANES results}. Our results are easy to interpret. For example, $[10.24, 10.70]$ (in the first row of Table \ref{tab:NHANES results}) for the effect of smoking status on the blood level of lead indicates that compared with no smoking cigarettes in the prior 30 days, smoking at least 10 cigarettes every day for the past 30 days will increase the blood level of cadmium by 10.24 nmol/L to 10.70 nmol/L. The interval $[10.24, 10.70]$ is informative about the sign of the treatment effect -- it contains only positive numbers so provides evidence that smoking causes an increase in cadmium. Compared to partial identification approaches using IVs, our method does not rely on the unconfoundedness of instruments. When implementing an IV-based partial identification approach without the assumption \citep{swanson2018partial}, we obtain the interval $[-52.24,68.84]$, which is excessively wide and thus uninformative. Additionally, for the bounds on the CATE, for instance,  $[9.09, 11.37]$ for the effect of smoking status on the blood level of cadmium conditional on males demonstrate that compared with no smoking cigarettes in the prior 30 days, smoking at least 10 cigarettes every day for the past 30 days will increase the blood level of cadmium by 9.09 nmol/L to 11.37 nmol/L within the male group.


\section{Discussion: Summary and Extension} \label{sec:discussion}

In this paper, we develop a differential effects approach to obtaining a flexible semi-parametric framework for estimating bounds for the ATE and CATE that can be plausible in a variety of settings.  Our bounds do not rely on very strong assumptions such as the no unmeasured confounding and the plausibility of key assumptions that the bounds rely on can be assessed under working models such as a 2-parameter logistic model.  The bounds can be informative such as the bounds we found for the effect of smoking on blood cadmium suggested that smoking increases blood cadmium. Our framework could be used to make inferences about individual treatment rules and personalized treatment since the optimal treatment rule can be estimated via partial identification of treatment effects \citep{pu2021estimating}. This could be a potential extension of our proposed framework.  

\section*{Acknowledgment}
This research was partially supported by NIH RF1AG063481 grant.

\section*{Supplementary Material}

The supplementary materials are divided into six parts. In Appendix I, II, III, IV, and V. In Appendix I, we provide a proof of Theorem \ref{thm: bound treatment effects} and in Appendix II, III, IV and V, we provide theoretical discussion of the monotonicity assumption, alternative bounding analysis, inference for the ATE and asymptotic properties of of $\hat \mu_1(x), \hat \mu_2 (x)$, respectively. Additional simulation results and case study results can be found in Appendix VI.

\bibliographystyle{apalike}
\bibliography{paper-ref}

\newpage

\begin{center}
    \Large Supplementary Materials for ``A Differential Effect Approach to Partial Identification of Treatment Effects"
\end{center}

\section*{Appendix I: Justification of Bounds on Treatment Effects}
To prove Theorem \ref{thm: bound treatment effects}, we need the following lemmas. 

\begin{lemma}[\cite{Esary1967}]
\label{lem: covariance_inequality}
Let $f(\cdot)$ and $g(\cdot)$ be functions with $K$ real-valued arguments, which are both non-decreasing in each of their arguments. If $U=\left(U_1, \ldots, U_K\right)$ is a multivariate random variable with $K$ mutually independent components, then $\operatorname{cov}\{f(U), g(U)\} \geq 0$.
\end{lemma}
Recall that the conditional counterfactual mean in arm $a$ is given by 
\begin{equation*}
    E[Y(a) \mid X = x] = \int E[Y \mid A = a, U = u, X = x] f(u \mid x) du,
\end{equation*}
where $f(u|x)$ is the conditional density function of $u$ given $x$. 
\begin{lemma}
\label{lem: outcome regression equivalence}
Under Assumptions \ref{assump: confoundness} and \ref{assump: ER},
\begin{equation*}
E[Y \mid Z_1 = z_1, Z_2 = z_2, U = u, X = x] = E[Y \mid Z_1 = z_1,  U = u, X = x],
\end{equation*}
for all $u, x$ and $z_1,z_2 \in \{0,1\}$.
\end{lemma}
\begin{proof}
By Assumption \ref{assump: confoundness},
\begin{align*}
E[Y(z_1) \mid  Z_2 = 1, U = u, X = x] &= E[Y(z_1) \mid Z_1 = z_1, Z_2 = 1, U = u, X = x] \\ &= E[Y \mid Z_1 = z_1, Z_2 = 1, U = u, X = x].
\end{align*}
Similarly, 
\begin{align*}
E[Y(z_1) \mid  Z_2 = 0, U = u, X = x] &= E[Y(z_1) \mid Z_1 = z_1, Z_2 = 0, U = u, X = x] \\ &= E[Y \mid Z_1 = z_1, Z_2 = 0, U = u, X = x].
\end{align*}
Then by Assumption \ref{assump: ER}, the left hand side quantities are equal, so $$E[Y \mid Z_1 = z_1, Z_2 = 1, U = u, X = x] = E[Y \mid Z_1 = z_1, Z_2 = 0, U = u, X = x],$$ which implies the result.
\end{proof}
\begin{lemma}
\label{lem: adjusted_for_z_2}
Under assumptions \ref{assumption: outcome_mono}, \ref{assumption: treatment_mono}, if the signs of the monotonicity match (i.e. both non-decreasing or non-increasing), 
\begin{equation*}
\begin{aligned}
    \sum_{z_2 \in \{0,1\}} E[Y \mid Z_1 = 1, Z_2 = z_2, X = x] P(Z_2 = z_2 \mid X = x) \geq E[Y(1) \mid X = x] \text{ and }  \\  \sum_{z_2 \in \{0,1\}} E[Y \mid Z_1 = 0, Z_2 = z_2, X = x] P(Z_2 = z_2 \mid X = x)  \leq E[Y(0) \mid X =x]. 
\end{aligned}
\end{equation*}
If the signs of monotonicity differ, 
\begin{equation*}
\begin{aligned}
    \sum_{z_2 \in \{0,1\}} E[Y \mid Z_1 = 1, Z_2 = z_2, X = x] P(Z_2 = z_2 \mid X = x) \leq E[Y(1) \mid X = x] \text{ and }  \\  \sum_{z_2 \in \{0,1\}} E[Y \mid Z_1 = 0, Z_2 = z_2, X = x] P(Z_2 = z_2 \mid X = x)  \geq E[Y(0) \mid X =x]. 
\end{aligned}
\end{equation*}
\end{lemma}
\begin{proof}
The proof uses similar ideas as in \cite{VanderWeele2008}. When the signs of monotonicity are the same, the conditional mean of $Y$ for $Z_1 = 1$ adjusted for the second treatment $Z_2$ is 
\begin{equation*}
\begin{aligned}
     \int & E[Y \mid Z_1 = 1, Z_2 = z_2, X = x] f(z_2 \mid x) dz_2 \\
     & = \int \int E[Y \mid Z_1 = 1, u,  z_2, 
 x] f(u \mid Z_1 = 1, z_2, x) f(z_2 \mid x) du dz_2  \\ &= \int \int E[Y \mid Z_1 = 1, u, z_2, x] P(Z_1 = 1 \mid u, z_2, x)/P(Z_1 = 1 \mid z_2, x) f(u \mid z_2, x) f(z_2 \mid x) du dz_2 \\ &= \int E\{ E[Y \mid Z_1 = 1, U, z_2, x] P(Z_1 = 1 \mid U, z_2, x)/P(Z_1 = 1 \mid z_2, x) \mid z_2, x \} f(z_2 \mid x)  dz_2 \\ & \geq \int E\{ E[Y \mid Z_1 = 1, U, z_2, x] \mid z_2, x \} f(z_2 \mid x)  dz_2 \\ &= \int \int E[Y \mid Z_1 = 1, u, z_2, x ] f(u\mid z_2, x)  f(z_2, x) du  dz_2 \\ &= \int \int E[Y \mid Z_1 = 1, u, z_2, x] f(u, z_2 \mid x) du  dz_2 \\  &=  E[Y(1) \mid X = x].
\end{aligned}
\end{equation*}
The inequality is due to Lemmas \ref{lem: covariance_inequality} and \ref{lem: outcome regression equivalence} combined with $E[Y \mid Z_1 = a, u, x]$ and $P[ Z_1 = 1 \mid u, z_2, x]$ either both non-decreasing or non-increasing in $u$ as well as $E_U\{P(Z_1 = 1 \mid U, z_2, x)/P(Z_1 = 1 \mid z_2, x) \mid z_2, x \} = 1$. 
Thus, we have established $\int E[Y \mid Z_1 = 1, z_2, x] f(z_2 \mid x) dz_2 \geq E[Y(1) \mid X = x] $ when the monotonicity signs match. When the signs differ, the $\geq$ sign in the above derivation turns to a $\leq$ sign, and so $\int E[Y \mid Z_1 = 1, z_2, x] f(z_2 \mid x) dz_2 \leq E[Y(1) \mid X = x] $ in that scenario. 

In the case for $E[Y(0) \mid X = x]$, the above derivation is almost the same, except with $Z_1 = 1$ replaced with $Z_1 = 0$. When the signs of the monotonicity of $E[Y \mid Z_1 = a, u, x]$ and $P[ Z_1 = 1 \mid u, z_2, x]$ match, the signs of the monotonicity of $E[Y \mid Z_1 = a, u, x]$ and $P[ Z_1 = 0 \mid u, z_2, x] = 1 - P[ Z_1 = 1 \mid u, z_2, x]$ disagree, so the $\geq$ sign in the above derivation would be a $\leq$ sign, implying $\int E[Y \mid Z_1 = 0, z_2, x] f(z_2 \mid x) dz_2 \leq E[Y(0) \mid X = x] $. On the other hand, when the signs of the monotonicity of $E[Y \mid Z_1 = a, u, x]$ and $P[ Z_1 = 1 \mid u, z_2, x]$ disagree, the signs of the monotonicity of $E[Y \mid Z_1 = a, u, z_2, x]$ and $P[ Z_1 = 0 \mid u, z_2, x] = 1 - P[ Z_1 = 1 \mid u, z_2, x]$ agree, so the $\geq$ sign in the above derivation would still be a $\geq$ sign, implying $\int E[Y \mid Z_1 = 0, z_2, x] f(z_2 \mid x) dz_2 \geq E[Y(0) \mid X = x] $. 

Thus, we have demonstrated that $$\int E[Y \mid Z_1 = 1, z_2, x] f(z_2 \mid x) dz_2 - \int E[Y \mid Z_1 = 0, z_2, x] f(z_2 \mid x) dz_2  \geq E[Y(1)-Y(0) \mid X = x]$$ when the signs agree and 
$$\int E[Y \mid Z_1 = 1, z_2, x] f(z_2 \mid x) dz_2 - \int E[Y \mid Z_1 = 0, z_2, x] f(z_2 \mid x) dz_2  \leq E[Y(1)-Y(0) \mid X = x]$$ when the signs disagree.
\end{proof}

\begin{proposition}
\label{prop: differential_effect_bound}
When Assumptions \ref{assump: confoundness}, \ref{assump: ER}, \ref{assump: positivity}, \ref{assumption: outcome_mono}, \ref{assumption: treatment_mono}, and \ref{assumption: z_2_influenced_more} hold,
$E[Y \mid Z_1 = 1, Z_2 = 0, X = x]$ and  $E[Y \mid Z_1 = 0, Z_2 = 1, X = x]$ serve as opposing bounds to $E[Y(1) \mid X = x]$ and $E[Y(0) \mid X = x]$ from Lemma \ref{lem: adjusted_for_z_2}, respectively.
\end{proposition}

\begin{proof}
Again, we use similar ideas as in \cite{VanderWeele2008}.
\begin{equation*}
\begin{aligned}
&E[Y \mid Z_1 = 1, Z_2 = 0, x] = \int E[Y \mid Z_1 = 1, Z_2 = 0, u, x] f(u \mid Z_1 = 1, Z_2 = 0, x) du \\  &= \int E[Y \mid Z_1 = 1, Z_2 = 0, u, x] P(Z_1 = 1,Z_2 = 0 \mid u, x)/P(Z_2 = 0, Z_1 = 1 \mid x)f(u \mid Z_1 = 1, Z_2 = 0, x) du \\  &= E\{E[Y \mid Z_1 = 1, Z_2 = 0, U, x] P(Z_2 = 0, Z_1 = 1\mid U, x)/P(Z_2 = 0, Z_1 = 1 \mid x) \mid x\} \\ &  \gtreqless 
 E_U\{E[Y \mid Z_1 = 1, Z_2 = 0, U, x ]\} \times E_U\{P(Z_2 = 0, Z_1 = 1 \mid U, x)/P(Z_1 = 1, Z_2 = 0 \mid x) \mid x\} \\ &= E_U\{E[Y \mid Z_1 = 1, Z_2 = 0, U, x]\} = 
 E_{U,Z_2}\{E[Y \mid Z_1 = 1, Z_2, U, x]\} = E[Y(1) \mid x].
\end{aligned} 
\end{equation*}
The inequality is due to Lemmas \ref{lem: covariance_inequality} and \ref{lem: outcome regression equivalence}. The second to last equality follows from Lemma \ref{lem: outcome regression equivalence}. The sign of the ambiguous $\gtreqless$ depends on the monotonicity signs of $P(Z_2 = 0, Z_1 = 1\mid u, x)$ and $E(Y \mid Z_1 = 1,  u, x)$. Note that by Assumption \ref{assumption: z_2_influenced_more}, $P(Z_2 = 0, Z_1 = 1\mid u, x)$ inherits the opposite sign of monotonicity as $P(Z_1 = 1\mid z_2, u, x)$. No matter the scenario, the $\gtreqless$ will be the opposite of the sign of comparing $\int E[Y \mid Z_1 = 1, Z_2 = z_2, X = x] f(z_2 \mid x) dz_2$ with $E[Y(1) \mid X = x]$, which is exactly what we set out to show. The steps for $E[Y \mid Z_1 = 0, Z_2 = 1, x]  \gtreqless E[Y(0) \mid x]$ are similar, and again, the sign of $\gtreqless$ will be opposite the sign of comparing $\int E[Y \mid Z_1 = 0, Z_2 = z_2, X = x] f(z_2 \mid x) dz_2$ with $E[Y(0) \mid X = x]$.
\end{proof}

\begin{proof}[Proof of Theorem \ref{thm: bound treatment effects}]
    Combining Lemma \ref{lem: adjusted_for_z_2} and Proposition \ref{prop: differential_effect_bound}, we have shown that we either have
$E[Y(1)- Y(0) \mid X = x] \in [\mu_1(x), \mu_2(x)]$ or $E[Y(1)- Y(0) \mid X = x] \in [\mu_2(x), \mu_1(x)]$, where 
\begin{align*}
&\mu_1(x) \equiv E[Y \mid Z_1 = 1, Z_2 = 0, x] 
 - E[Y \mid Z_1 = 0, Z_2 = 1, x] 
 \\
&\mu_2(x) \equiv \int E[Y \mid Z_1 = 1, Z_2 = z_2, X = x] f(z_2 \mid x) dz_2 - \int E[Y \mid Z_1 = 0, Z_2 = z_2, X = x] f(z_2 \mid x) dz_2 \\
& = \mathbb{E}\left[\frac{Z_1 Y}{\mathbb{E} [Z_1|Z_2,X] } \Bigg| X=x\right] - \mathbb{E}\left[\frac{(1 -Z_1) Y}{\mathbb{E} [1 - Z_1|Z_2,X] } \Bigg| X=x\right].
\end{align*}
These bounds exactly match the CATE bounds (up to $\mu_2(x)$ being written in IPW form rather than outcome regression form). ATE bounds are obtained by taking the expectation with respect to the distribution of $X$. 
\end{proof}

\section*{Appendix II: Justification of Assumption \ref{assumption: z_2_influenced_more}}

To illustrate Assumption \ref{assumption: z_2_influenced_more}, we consider verifying it under a simple model. Specifically, we consider  the 2-parameter logistic model. Under the the 2-parameter logistic model , $Z_1 \perp Z_2 \mid U, X$, and 
\begin{align*}
&\mathbb{P}(Z_1 = a, Z_2 = b | U = u, X = x) \\
    &= \frac{\exp \{a  (\alpha_1 u + \beta_1 x) + b  (\alpha_2 u + \beta_2 x) \} }{1 + \exp\{\alpha_1 u + \beta_1 x \} + \exp\{\alpha_2 u + \beta_2 x \} + \exp\{\alpha_1 u + \beta_1 x + \alpha_2 u + \beta_2 x \}}.
\end{align*}
Here, $x$ (which includes an intercept) and the $\beta$'s can be though of as vectors, with $\beta_1 x$ and $\beta_2 x$ denoting inner products. Then, 

\begin{align*}
\mathbb{P}(Z_1 = 0, Z_2 = 1 \mid U = u, X = x) &= \frac{\exp \{(\alpha_2 u + \beta_2 x) \} }{1 + \exp\{\alpha_1 u + \beta_1 x \} + \exp\{\alpha_2 u + \beta_2 x \} + \exp\{\alpha_1 u + \beta_1 x + \alpha_2 u + \beta_2 x \}},
\\
\mathbb{P}(Z_1 = 1, Z_2 = 0 \mid U = u, X = x) &= \frac{\exp \{(\alpha_1 u + \beta_1 x)  \} }{1 + \exp\{\alpha_1 u + \beta_1 x \} + \exp\{\alpha_2 u + \beta_2 x \} + \exp\{\alpha_1 u + \beta_1 x + \alpha_2 u + \beta_2 x \}}.
\end{align*}

When $U$ is a continuous variable, a derivative calculation yields
\begin{align*}
&\frac{d}{du} \mathbb{P}(Z_1 = 0, Z_2 = 1 \mid U = u, X = x) \\ & = \frac{\exp(\alpha_2 u + \beta_2 x) \left[ \alpha_2 - \alpha_1 \exp(\alpha_1 u + \beta_1 x) - \alpha_1 \exp(\alpha_1 u + \beta_1 x + \alpha_2 u + \beta_2 x) \right]}{\left( 1 + \exp(\alpha_1 u + \beta_1 x) + \exp(\alpha_2 u + \beta_2 x) + \exp(\alpha_1 u + \beta_1 x + \alpha_2 u + \beta_2 x) \right)^2}, \\
&\frac{d}{du} \mathbb{P}(Z_1 = 1, Z_2 = 0 \mid U = u, X = x) \\ & = \frac{\exp(\alpha_1 u + \beta_1 x) \left[ \alpha_1 - \alpha_2 \exp(\alpha_2 u + \beta_2 x) - \alpha_2 \exp(\alpha_1 u + \beta_1 x + \alpha_2 u + \beta_2 x) \right]}{\left( 1 + \exp(\alpha_1 u + \beta_1 x) + \exp(\alpha_2 u + \beta_2 x) + \exp(\alpha_1 u + \beta_1 x + \alpha_2 u + \beta_2 x) \right)^2}.
\end{align*}

The first quantity will be non-negative as long as 
\begin{align*}
&\alpha_2 - \alpha_1 \exp(\alpha_1 u + \beta_1 x) - \alpha_1 \exp(\alpha_1 u + \beta_1 x + \alpha_2 u + \beta_2 x) \geq 0 \iff \\ & \alpha_2/\alpha_1 \geq \exp(\alpha_1 u + \beta_1 x) + \exp(\alpha_1 u + \beta_1 x + \alpha_2 u + \beta_2 x).
\end{align*}

The second quantity will be non-positive as long as 
\begin{align*}
&\alpha_1 - \alpha_2 \exp(\alpha_2 u + \beta_2 x) - \alpha_2 \exp(\alpha_1 u + \beta_1 x + \alpha_2 u + \beta_2 x) \leq 0 \iff \\ & \alpha_2/\alpha_1 \geq \exp(\alpha_2 u + \beta_2 x) + \exp(\alpha_1 u + \beta_1 x + \alpha_2 u + \beta_2 x).
\end{align*}

In the context of the real data application, we fit a Bayesian model, assuming $U$ is binary rather than continuous, and fit an analogous model with 4 items (smoking, hard drug usage, high fat intake, alcohol abuse). The 4 items are modeled as conditionally independent given $U, X$. Then based on Bayesian posterior estimates for the $\alpha$ and $\beta$ parameters, we check if $\mathbb{P}(Z_1 = 0, Z_2 = 1 \mid U = 1, X = x) \geq \mathbb{P}(Z_1 = 0, Z_2 = 1 \mid U = 0, X = x)$ and whether $\mathbb{P}(Z_1 = 1, Z_2 = 0 \mid U = 1, X = x) - \mathbb{P}(Z_1 = 1, Z_2 = 0 \mid U = 0, X = x) \leq 0$. The specific model is as follows:

\begin{equation*}
\begin{aligned}
U \sim \text{Bern}(0.2), \\ 
\alpha_1, \ldots, \alpha_4 \stackrel{\text{ind}}{\sim} N(0.2, 2 \times I_d), \\ \beta_1, \ldots, \beta_4 \stackrel{\text{ind}}{\sim} N(0, 2 \times I_d).
\end{aligned}
\end{equation*}

Since the binary $U$ is latent, we instead marginalize over the $U$ that has a fixed success probability, and is modeled as independent from all other data. The marginalized (over $U$) likelihood from a single observation is as follows:

\begin{align*}
 l(z_1,\ldots z_4; \mathbf{\alpha}, \mathbf{\beta}) &= 0.2 \times \prod_{i=1}^4 \left( \frac{\exp(\alpha_i + \beta_i x)}{1 + \exp(\alpha_i + \beta_i x)}\right)^{z_i} \left( \frac{1}{1 + \exp(\alpha_i + \beta_i x)}\right)^{1 - z_i} \\ &+ 0.8 \times \prod_{i=1}^4 \left( \frac{\exp(\beta_i x)}{1 + \exp(\alpha_i + \beta_i x)}\right)^{z_i} \left( \frac{1}{1 + \exp(\alpha_i + \beta_i x)}\right)^{1 - z_i}   
\end{align*}
We fit the model using the \texttt{rstan} software, with the standard, recommended 1000 warm-up iterations and 1000 sampling iterations on 4 chains. All parameters were constrained to lie between -3 and 3. We take the mean of the 1000 posterior samples as a point estimate for each of the $\alpha$ and $\beta$ parameters. At each value of covariates $X = x$ in the data, we compute estimated $\mathbb{P}(Z_1 = 0, Z_2 = 1 \mid U = u, X = x)$ and $\mathbb{P}(Z_1 = 0, Z_2 = 1 \mid U = u, X = x)$ values (plugging in the posterior means of the parameters). We found that the estimated probabilities followed the desired inequalities $\mathbb{P}(Z_1 = 0, Z_2 = 1 \mid U = 1, X = x) \geq \mathbb{P}(Z_1 = 0, Z_2 = 1 \mid U = 0, X = x)$ and $\mathbb{P}(Z_1 = 1, Z_2 = 0 \mid U = 1, X = x) - \mathbb{P}(Z_1 = 1, Z_2 = 0 \mid U = 0, X = x) \leq 0$ at all values of the covariates in the dataset. This suggests that the Assumption \ref{assumption: z_2_influenced_more} is plausible in the data application.

\begin{table}[!tb]
    \centering
    \begin{tabular}{|ccc|}
    \hline
        Item (Binary) & $\alpha$ (Discrimination) & Standard Deviation\\ \hline 
        Smoking Status & 0.003   & 1.501  \\ 
        Hard Drug Usage Status & 0.299  & 1.534  \\
        Unhealthy Fat Intake & 1.001 & 1.693\\ 
        Alcohol Abuse & 0.753 & 1.522\\ \hline
    \end{tabular}
    \caption{ Summary table of fitting a Bayesian 2-parameter logistic IRT model to the NHANES dataset with 4 items. The second column collects the posterior means of the $\alpha$ parameters, and the third column the standard deviation of the 1000 posterior draws.}
    \label{tab: checking assumption}
\end{table}

\section*{Appendix III: Alternative Bounding Analysis via Differential Effects}

Alternatively, to bound the ATE and CATE under different set of assumptions, we provide relaxed exclusion restriction and modified monotonicity assumption below.

\begin{assumption}[Relaxed Exclusion Restriction] \label{assump: RER}
 Let $ \mathbb{E}[Y(z_1) | Z_2 = 1, X, U] + \epsilon_1 =  \mathbb{E}[Y(z_1) | Z_2 = 0, X, U] + \epsilon_2 = \mathbb{E}[Y(z_1) |  X, U]$ for some constant $\epsilon_1$ and $\epsilon_2$, then 
 either (a) $\epsilon_1 \geq 0, \epsilon_2 \leq 0$  or (b) $\epsilon_1 \leq 0, \epsilon_2 \geq 0$ holds  for every $X, U,$ and $z_1 \in \{0,1\}.$
\end{assumption}

It is clear to see that $\epsilon_1 \epsilon_2 \leq 0$ in Assumption \ref{assump: RER}. With Assumption \ref{assump: RER}, we allow for a direct effect of $Z_2$ on outcome which enlarges the scope of candidates for $Z_2$. However, we require this direct effect to be constant across covariate $X$ and $U$, since otherwise no constraint is put on $(Y, Z_2)$. Assumption \ref{assump: RER} (a) implies that there is a direct negative effect of $Z_2$ on $Y$ while (b) implies that there is a direct positive effect of $Z_2$ on $Y$. Assumption \ref{assump: RER} also says that $\mathbb{E}[Y(z_1)|X,U]$ has the same functional form as $\mathbb{E}[Y(z_1)|Z_2 = b, X, U]$ up to a constant (i.e., $Z_2 \indep (X,U) $).

\subsection*{Bounding ATE}

To bound the ATE, let $\mathbb{E}[f_a(X, U)|Z_1 = a, Z_2 = b, X = x ]  =  h(a, b, x)$ where $ f_{a}(X,U) = \mathbb{E}[Y(a)|U,X] - \mathbb{E}[Y(a)] $ for $a \in \{0,1\} $. We then make the following conditional monotonicity assumptions. 

\begin{assumption}[Conditional Monotonicity for ATE] \label{assump: monotone conditioned ATE}
For any fixed $x \in \mathbb{R}^l$, either (a) $h(0,0,x) \leq h(1,0,x) \leq h(0,1,x) \leq h(1,1,x)$ or (b) $h(0,0,x) \geq h(1,0,x) \geq h(0,1,x) \geq h(1,1,x)$ holds. 
\end{assumption}

Again, we consider the meaning of Assumption \ref{assump: monotone conditioned ATE} in the context of the 2-parameter logistic model  from  item response theory (IRT) \citep{hambleton2013item} where the item responses are the two treatments $Z_1$ and $Z_2$:

\begin{align*}
    \mathbb{P} (Z_1 = 1 |U = u, X = x) = \frac{\exp\{\alpha_1 u   + \beta_1 x \} }{1 + \exp\{ \alpha_1 u + \beta_1 x \} },
\end{align*}
and
\begin{align*}
        \mathbb{P} (Z_2 = 1 | U = u, X = x) = \frac{\exp\{\alpha_2 u  + \beta_2 x \} }{1 + \exp\{ \alpha_2 u  + \beta_2 x \} },
\end{align*}
 where $u$ is the latent variable, $\alpha_1$ and $\alpha_2$ are discrimination parameters while $\beta_1$ and $\beta_2$ are coefficients of covariate $x$. In our context, the latent variable $u$ is the unmeasured confounding, and $x$ is age; $\alpha_1, \alpha_2$  determine the rate at which the probability of taking the corresponding treatment $Z_1$ or $Z_2$ changes as a function of the level of unmeasured confounding. By the local independence assumption from IRT, that is, item responses are independent conditional on  the latent variable, we have
\begin{equation} 
    \begin{split}
           &\qquad \mathbb{P}(Z_1 = a, Z_2 = b | U = u, X = x) \\
    &= \frac{\exp \{a  (\alpha_1 u + \beta_1 x) + b  (\alpha_2 u + \beta_2 x) \} }{1 + \exp\{\alpha_1 u + \beta_1 x \} + \exp\{\alpha_2 u + \beta_2 x \} + \exp\{\alpha_1 u + \beta_1 x + \alpha_2 u + \beta_2 x \}}.
    \end{split}
\end{equation}
Hence, 
\begin{align*}
     h(a, b, x) = \frac{ \int_U f_a (x, u)\mathbb{P}(Z_1 = a, Z_2 = b | U = u, X = x) g (u|x)  du  }{  \int_U \mathbb{P} (Z_1 = a, Z_2 = b | U = u, X = x) g (u|x)  du  }
\end{align*}
with $ g(u|x)$ being density functions of $U$ given $X$. Then $h(0,0,x) \leq h(1,0,x) \leq h(0,1,x) \leq h(1,1,x)$ if $\alpha_2 \geq \alpha_1 > 0$ and $\mathbb{E}[Y(a)|X = x, U = u] - \mathbb{E}[Y(a)]$ is non-decreasing in $u$, and $a$ given $x$. And $h(0,0,x) \geq h(1,0,x) \geq h(0,1,x) \geq h(1,1,x)$ if $ \alpha_2 \leq \alpha_1 < 0 $ and $\mathbb{E}[Y(a)|X = x, U = u] - \mathbb{E}[Y(a)]$ is non-increasing in $u$ and $a$ given $x$. 

To understand the notation and assumptions better, we consider the study of the effect of smoking on the blood level of cadmium using the NHANES dataset. In our study, $Y$ represents the blood level of cadmium, $Z_1$ is smoking, $Z_2$ is past hard drug usage (we exclude people who are currently using hard drugs from the study), $X$ are observed confounders such as age, and $U$ is unmeasured confounding such as other health risky behaviors. In this case, assumption \ref{assump: monotone conditioned ATE} holds under the logistic mixed effect model if 1.) hard drug usage has a higher discrimination parameter (i.e. a higher rate at which the probability of hard drug usage status changes given a certain level of unmeasured confounding) than smoking, and 2.) other health risky behaviors ($U$), age, as well as smoking ($Z_1$), have a non-decreasing effect on the blood level of cadmium ($Y$) for any age (i.e. $\mathbb{E}[Y(z_1)|X = x, U = u]$ is non-decreasing in  $u$ and $z_1$ given any $x$). In general, assumption \ref{assump: monotone conditioned ATE} could hold under a broader class of models \citep{holland1986conditional}. Recall that we define $\mu_1 = \mathbb{E}_X [\mu_1(X)], \mu_2 = \mathbb{E}_X[ \mu_2(X)], \tau^+ = \max\{\mu_1,\mu_2\}$ and $\tau^- = \min\{\mu_1,\mu_2\}$, then we have the bounding analysis result for the ATE. 

\begin{theorem}[Identification of the ATE] \label{thm: bound tau}
$\tau \in [\tau^-, \tau^+]$ if

\begin{itemize}
    \item[1.)] assumption \ref{assump: confoundness}, \ref{assump: ER}, \ref{assump: positivity}, and \ref{assump: monotone conditioned ATE} (a) hold, or
    \item[2.)] assumption \ref{assump: confoundness}, \ref{assump: ER}, \ref{assump: positivity}, and \ref{assump: monotone conditioned ATE} (b) hold, or
    \item[3.)] assumption \ref{assump: confoundness}, \ref{assump: RER} (a), \ref{assump: positivity}, and \ref{assump: monotone conditioned ATE} (b) hold, or
    \item[4.)] assumption \ref{assump: confoundness}, \ref{assump: RER} (b), \ref{assump: positivity}, and \ref{assump: monotone conditioned ATE} (a) hold. 
\end{itemize}
Additionally, $\tau = \mu_1 = \mu_2$ if assumption \ref{assump: confoundness}, \ref{assump: ER}, \ref{assump: positivity} hold and $h(0,0,x) = h(1,0,x) = h(0,1,x) = h(1,1,x)$ for every $x \in \mathbb{R}^l$. 

\end{theorem}

\begin{proof}
    We first consider cases of 1.) and 2.). For differential effect, we have
    \begin{align*}
        & \quad \mathbb{E}_X \left\{ \mathbb{E} [Y|Z_1 = 1, Z_2 = 0, X] \right\} = \mathbb{E}[Y|Z_1 = 1, Z_2 = 0] \\
        &= \mathbb{E} [Z_1 Y(1) + (1 - Z_1) Y(0) |Z_1 = 1, Z_2 = 0 ] \quad \text{(by the consistency assumption)}   \\
        & = \mathbb{E} [Y(1) | Z_1 = 1, Z_2 = 0 ] \\
        & = \mathbb{E} \left\{ \mathbb{E}[Y(1) | Z_1 = 1, Z_2 = 0, X, U   ] | Z_1 = 1, Z_2 = 0     \right\} \quad \text{(by the law of total expectation)} \\
        & = \mathbb{E} \left\{ \mathbb{E}[Y(1) | Z_2 = 0, X, U   ] | Z_1 = 1, Z_2 = 0     \right\} \quad \text{(by assumption \ref{assump: confoundness})} \\
        & = \mathbb{E} \left\{ \mathbb{E}[Y(1) | X, U   ] | Z_1 = 1, Z_2 = 0     \right\} \quad \text{(by assumption \ref{assump: ER})} \\
        & = \mathbb{E}[Y(1)] + \mathbb{E} \left\{ f_1 (X, U) | Z_1 = 1, Z_2 = 0     \right\}. 
    \end{align*}
Similarly, we have $ \mathbb{E}_X\left\{ \mathbb{E}[Y|Z_1 = 0, Z_2 = 1, X]  \right\} =  \mathbb{E}[Y|Z_1 = 0, Z_2 = 1] = \mathbb{E}[Y(0)] + \mathbb{E} \left\{ f_0 (X, U) | Z_1 = 0, Z_2 = 1     \right\} $. For IPW, we have
\begin{align*}
    & \quad \mathbb{E}_X \left\{ \mathbb{E}\left[ \frac{Z_1 Y}{ \mathbb{E} [Z_1 | Z_2, X] } \mid X  \right] \right\} = \mathbb{E}\left[ \frac{Z_1 Y}{ \mathbb{E} [Z_1 | Z_2, X] }  \right] \\
    &= \mathbb{E}\left[ \frac{Z_1 Y(1)}{ \mathbb{E} [Z_1 | Z_2, X] }  \right] \quad \text{(by the consistency assumption)}   \\
    & =  \mathbb{E} \Big[ \mathbb{E} \{\frac{Z_1 Y(1)}{\mathbb{E}[Z_1 | Z_2, X] } | Z_2, X, U   \}      \Big] \quad \text{(by the law of total expectation)} \\
    & = \mathbb{E} \Big[ \frac{1}{\mathbb{E}[Z_1 | Z_2, X]} \mathbb{E} \{ Z_1 | Z_2, X, U   \} \mathbb{E} \{ Y(1) | Z_2, X, U   \}     \Big] \quad \text{(by assumption \ref{assump: confoundness})} \\
    & = \mathbb{E} \Big[ \frac{1}{\mathbb{E}[Z_1 | Z_2, X]} \mathbb{E} \{ Z_1 | Z_2, X, U   \} \mathbb{E} \{ Y(1) | X, U   \}     \Big] \quad \text{(by assumption \ref{assump: ER})} \\
    & = \mathbb{E} \Big[ \frac{1}{\mathbb{E}[Z_1 | Z_2, X]} \mathbb{E} \{ Z_1 | Z_2, X, U   \} \{ \mathbb{E}[Y(1)] + f_1 (X, U)   \}  \Big]  \\
    & = \mathbb{E} \Big[ \mathbb{E} \{ \frac{1}{\mathbb{E}[Z_1 | Z_2, X]} Z_1 \{ \mathbb{E}[Y(1)] + f_1 (X, U)   \} | Z_2, X, U \} \Big] \quad \text{(by the law of total expectation)} \\
    & = \mathbb{E} \Big[  \frac{1}{\mathbb{E}[Z_1 | Z_2, X]} Z_1 \{ \mathbb{E}[Y(1)] + f_1 (X, U)   \}  \Big] \\
    & = \mathbb{E}[Y(1)] + \mathbb{E} \Big[ \frac{1}{\mathbb{E}[Z_1|Z_2, X] } Z_1 f_1 (X,U)  \Big] \\
    & = \mathbb{E}[Y(1)] + \mathbb{E} \Big[ \frac{1}{\mathbb{E}[Z_1|Z_2, X] } Z_1 \mathbb{E}[f_1 (X,U) | Z_1, Z_2, X ] \Big] \quad \text{(by the law of total expectation)}\\
    & = \mathbb{E}[Y(1)] + \mathbb{E} \Big\{ \mathbb{E} \Big[ \frac{1}{\mathbb{E}[Z_1|Z_2, X] } Z_1 \mathbb{E}[f_1 (X,U) | Z_1, Z_2, X ] \mid X \Big] \Big\} \quad \text{(by the law of total expectation)}\\
    & = \mathbb{E}[Y(1)] + \mathbb{E} \Big\{ \sum_{a \in \{0,1 \} } \sum_{b \in \{0,1\}} \frac{a}{\mathbb{E}[Z_1|Z_2 = b, X] } \mathbb{E}[f_1(X,U)|Z_1 = a, Z_2 = b, X]  \mathbb{P}(Z_1 = a, Z_2 = b | X) \Big\} \\
    & = \mathbb{E}[Y(1)] + \mathbb{E} \Big\{  \sum_{b \in \{0,1\}} \frac{1}{\mathbb{E}[Z_1|Z_2 = b, X] } \mathbb{E}[f_1(X,U)|Z_1 = 1, Z_2 = b, X]  \mathbb{P}(Z_1 = 1, Z_2 = b | X) \Big\} \\
        & = \mathbb{E}[Y(1)] + \mathbb{E} \Big\{  \sum_{b \in \{0,1\}} \frac{1}{\mathbb{E}[Z_1|Z_2 = b, X] } \mathbb{E}[f_1(X,U)|Z_1 = 1, Z_2 = b, X]  \mathbb{P}(Z_1 = 1 \mid Z_2 = b,  X) \mathbb{P} (Z_2 = b \mid X) \Big\} \\
    & = \mathbb{E}[Y(1)] +  \sum_{b \in \{0,1\}} \mathbb{E} \Big\{  \mathbb{E}[f_1(X,U)|Z_1 = 1, Z_2 = b, X]  \mathbb{P}(Z_2 = b \mid X) \Big\}. 
\end{align*}
Similarly, we have
\begin{align*}
  & \quad \mathbb{E}_X \left\{ \mathbb{E}\left[\frac{(1 -Z_1) Y}{\mathbb{E} [1 - Z_1|Z_2,X] } \mid X \right] \right\}=  \mathbb{E}\left[\frac{(1 -Z_1) Y}{\mathbb{E} [1 - Z_1|Z_2,X] }\right] \\ 
  &=  \mathbb{E}[Y(0)] +  \sum_{b \in \{0,1\}} \mathbb{E} \Big\{  \mathbb{E}[f_0(X,U)|Z_1 = 0, Z_2 = b, X]  \mathbb{P}(Z_2 = b \mid X) \Big\}. 
\end{align*}
Hence, 
\begin{align*}
    & \qquad \tau + \mathbb{E}[f_1(X,U)|Z_1 = 1, Z_2 = 0] - \mathbb{E}[f_0(X,U)|Z_1 = 0, Z_2 = 1] \\
& = \mathbb{E}[Y|Z_1 = 1, Z_2 = 0] -  \mathbb{E}[Y|Z_1 = 0, Z_2 = 1] = \mu_1,
\end{align*}
and
\begin{align*}
    & \qquad \tau +  \sum_{b \in \{0,1\}} \mathbb{E} \Big\{  \left(\mathbb{E}[f_1(X,U)|Z_1 = 1, Z_2 = b, X] -   \mathbb{E}[f_0(X,U)|Z_1 = 0, Z_2 = b, X] \right) \mathbb{P}(Z_2 = b \mid X) \Big\} \\
    &= \mathbb{E}\left[\frac{Z_1 Y}{\mathbb{E} [Z_1|Z_2,X] }\right] - \mathbb{E}\left[\frac{(1 -Z_1) Y}{\mathbb{E} [1 - Z_1|Z_2,X] }\right] = \mu_2.
\end{align*}
It is clear to see that $\tau = \mu_1 = \mu_2$ if $h(0,0,x) = h(1,0,x) = h(0,1,x) = h(1,1,x)$. $\mu_1 \leq \tau \leq \mu_2$ if $h(0,0,x) \leq h(1,0,x) \leq h(0,1,x) \leq h(1,1,x)$. And $\mu_2 \leq \tau \leq \mu_1$ if $h(0,0,x) \geq h(1,0,x) \geq h(0,1,x) \geq h(1,1,x)$. For cases 3.) and 4.), under assumption \ref{assump: RER},  we have
\begin{align*}
    \mathbb{E}[Y(z_1) | Z_2 = 1, X, U] + \epsilon_1 =  \mathbb{E}[Y(z_1) | Z_2 = 0, X, U] + \epsilon_2 = \mathbb{E}[Y(z_1) |  X, U].
\end{align*}
Then by similar algebra as before, we have 
\begin{align*}
    & \qquad \tau + \mathbb{E}[f_1(X,U)|Z_1 = 1, Z_2 = 0] - \mathbb{E}[f_0(X,U)|Z_1 = 0, Z_2 = 1] \\
& = \mathbb{E}[Y|Z_1 = 1, Z_2 = 0] -  \mathbb{E}[Y|Z_1 = 0, Z_2 = 1] + (\epsilon_1 - \epsilon_2) = \mu_1 + (\epsilon_1 - \epsilon_2),
\end{align*}
and
\begin{align*}
    & \qquad \tau + \sum_{b \in \{0,1\}} \mathbb{E} \Big\{  \left(\mathbb{E}[f_1(X,U)|Z_1 = 1, Z_2 = b, X] -   \mathbb{E}[f_0(X,U)|Z_1 = 0, Z_2 = b, X] \right) \mathbb{P}(Z_2 = b \mid X) \Big\} \\
    &= \mathbb{E}\left[\frac{Z_1 Y}{\mathbb{E} [Z_1|Z_2,X] }\right] - \mathbb{E}\left[\frac{(1 -Z_1) Y}{\mathbb{E} [1 - Z_1|Z_2,X] }\right]  = \mu_2.
\end{align*}
It is clear to see that $\mu_1 + (\epsilon_1 - \epsilon_2) \leq \tau \leq \mu_2$ if $h(0,0,x) \leq h(1,0,x) \leq h(0,1,x) \leq h(1,1,x)$ and $\epsilon_1 - \epsilon_2 \leq 0$ (assumption \ref{assump: RER} (b)). And $\mu_2 \leq \tau \leq \mu_1 + (\epsilon_1 - \epsilon_2)$ if $h(0,0,x) \geq h(1,0,x) \geq h(0,1,x) \geq h(1,1,x)$ and $\epsilon_1 - \epsilon_2 \geq 0$ (assumption \ref{assump: RER} (a)). 
\end{proof}

\subsection*{Bounding CATE}

To bound the CATE, let $\mathbb{E}[k_a(X, U)|Z_1 = a, Z_2 = b, X = x ]  =  l(a, b, x)$ where $ k_{a}(X,U) = \mathbb{E}[Y(a)|U,X] - \mathbb{E}[Y(a)|X] $ for $a \in \{0,1\} $. Note that $k_a (X, U)$ is different from $f_a (x,u)$ in the second term. However, $f_a (x,u)$ and $k_a(x,u)$ exhibit the same pattern in $u$ for fixed $a$ and $x$. Then we have the following conditional monotonicity assumption for CATE.  

\begin{assumption}[Conditional Monotonicity for CATE] \label{assump: monotone conditioned CATE}
For any fixed $x \in \mathbb{R}^l$, either (a) $l(0,0,x) \leq l(1,0,x) \leq l(0,1,x) \leq l(1,1,x)$ or (b) $l(0,0,x) \geq l(1,0,x) \geq l(0,1,x) \geq l(1,1,x)$ holds. 
\end{assumption}

We use an example to demonstrate the insight behind assumption \ref{assump: monotone conditioned CATE}. Similar as assumption \ref{assump: monotone conditioned ATE}, consider the model in (\ref{equ:rasch model}), we have 
\begin{align*}
l(a,b,x) = \frac{ \int_U k_a (x,u) \mathbb{P}(Z_1 = a, Z_2 = b | U = u, X = x) g (u|x) p(x) du  }{  \int_U \mathbb{P} (Z_1 = a, Z_2 = b | U = u, X = x) g (u|x) p(x) du  }.
\end{align*}
Again, $l(0,0,x) \leq l(1,0,x) \leq l(0,1,x) \leq l(1,1,x)$ if $0 < \alpha_1 \leq \alpha_2$ and $\mathbb{E}[Y(a)|U = u, X=x]- \mathbb{E}[Y(a)|X]$ is non-decreasing in both $u$ and $a$ for any fixed $x \in \mathbb{R}^l$. And $l(0,0,x) \geq l(1,0,x) \geq l(0,1,x) \geq l(1,1,x)$ if $0 > \alpha_1 \geq \alpha_2$ and $\mathbb{E}[Y(a)|U = u, X=x]- \mathbb{E}[Y(a)|X]$ is non-increasing in both $u$ and $a$ for any fixed $x \in \mathbb{R}^l$. 

In the motivating example, similar to assumption \ref{assump: monotone conditioned ATE}, assumption \ref{assump: monotone conditioned CATE} holds under the logistic mixed effect model if 1.) hard drug usage status has a higher discrimination parameter (i.e. a higher rate at which the probability of hard drug usage status changes given a certain level of unmeasured confounding) than smoking, and 2.) other health risky behavior ($U$), as well as smoking ($Z_1$), have a non-decreasing effect on the blood level of cadmium ($Y$) on average for all levels of the observed covariates such as age ($X$) (i.e. $\mathbb{E}[Y(z_1)|U = u, X=x] - \mathbb{E}[Y(a)|X]$ is non-decreasing in both $u$ and $z_1$ for any fixed $x \in \mathbb{R}^+$). With the assumptions listed above, the main result for the identification of the CATE can be stated.

\begin{theorem}[Identification of the CATE] \label{thm: bound tau(x)} $\tau (x) \in [\tau^- (x), \tau^+ (x)]$ for any fixed $x \in \mathbb{R}^l$ if
\begin{itemize}
    \item[1.)] assumption \ref{assump: confoundness}, \ref{assump: ER}, \ref{assump: positivity}, and \ref{assump: monotone conditioned CATE} (a) hold, or
    \item[2.)] assumption \ref{assump: confoundness}, \ref{assump: ER}, \ref{assump: positivity}, and \ref{assump: monotone conditioned CATE} (b) hold, or
    \item[3.)] assumption \ref{assump: confoundness}, \ref{assump: RER} (a), \ref{assump: positivity}, and \ref{assump: monotone conditioned CATE} (b) hold, or
    \item[4.)] assumption \ref{assump: confoundness}, \ref{assump: RER} (b), \ref{assump: positivity}, and \ref{assump: monotone conditioned CATE} (a) hold. 
\end{itemize}
Additionally, $\tau (x) = \mu_1 (x) = \mu_2 (x)$ if assumption \ref{assump: confoundness}, \ref{assump: ER}, \ref{assump: positivity} hold and $h(0,0 ,x) = h(1,0, x) = h(0,1, x) = h(1,1, x)$. 
\end{theorem}

\begin{proof}
    Similar to the proof of Theorem \ref{thm: bound tau},  we first consider cases of 1.) and 2.). For the conditional differential effect, we have
    \begin{align*}
        & \quad \mathbb{E}[Y|Z_1 = 1, Z_2 = 0, X] \\
        & = \mathbb{E} [Z_1 Y(1) + (1 - Z_1) Y(0) |Z_1 = 1, Z_2 = 0,X ] \quad \text{(by the consistency assumption)}   \\
        & = \mathbb{E} [Y(1) | Z_1 = 1, Z_2 = 0, X ] \\
        & = \mathbb{E} \left\{ \mathbb{E}[Y(1) | Z_1 = 1, Z_2 = 0, X, U   ] | Z_1 = 1, Z_2 = 0, X     \right\} \quad \text{(by the law of total expectation)} \\
        & = \mathbb{E} \left\{ \mathbb{E}[Y(1) | Z_2 = 0, X, U   ] | Z_1 = 1, Z_2 = 0, X     \right\} \quad \text{(by assumption \ref{assump: confoundness})} \\
        & = \mathbb{E} \left\{ \mathbb{E}[Y(1) | X, U   ] | Z_1 = 1, Z_2 = 0, X     \right\} \quad \text{(by assumption \ref{assump: ER})} \\
        & = \mathbb{E}[Y(1) | X] + \mathbb{E} \left\{ k_1 (X, U) | Z_1 = 1, Z_2 = 0, X     \right\}. 
    \end{align*}
Similarly, we have 
\begin{align*}
     \mathbb{E}[Y|Z_1 = 0, Z_2 = 1, X] = \mathbb{E}[Y(0)|X] + \mathbb{E} \left\{ k_0 (X, U) | Z_1 = 0, Z_2 = 1, X     \right\}
\end{align*}. 
For the IPW, we have
\begin{align*}
    &\quad \mathbb{E}\left[ \frac{Z_1 Y}{ \mathbb{E} [Z_1 | Z_2, X] } \bigg| X  \right] \\
    &= \mathbb{E}\left[ \frac{Z_1 Y(1)}{ \mathbb{E} [Z_1 | Z_2, X] } \bigg| X  \right] \quad \text{(by the consistency assumption)}   \\
    & =  \mathbb{E} \Big[ \mathbb{E} \{\frac{Z_1 Y(1)}{\mathbb{E}[Z_1 | Z_2, X] } | Z_2, X, U   \}  \bigg| X    \Big] \quad \text{(by the law of total expectation)} \\
    & = \mathbb{E} \Big[ \frac{1}{\mathbb{E}[Z_1 | Z_2, X]} \mathbb{E} \{ Z_1 | Z_2, X, U   \} \mathbb{E} \{ Y(1) | Z_2, X, U   \}   \big| X  \Big] \quad \text{(by assumption \ref{assump: confoundness})} \\
    & = \mathbb{E} \Big[ \frac{1}{\mathbb{E}[Z_1 | Z_2, X]} \mathbb{E} \{ Z_1 | Z_2, X, U   \} \mathbb{E} \{ Y(1) | X, U   \}    \big| X \Big] \quad \text{(by assumption \ref{assump: ER})} \\
    & = \mathbb{E} \Big[ \frac{1}{\mathbb{E}[Z_1 | Z_2, X]} \mathbb{E} \{ Z_1 | Z_2, X, U   \} \{ \mathbb{E}[Y(1)|X] + k_1 (X, U)   \} \big| X  \Big]  \\
    & = \mathbb{E} \Big[ \mathbb{E} \{ \frac{1}{\mathbb{E}[Z_1 | Z_2, X]} Z_1 \{ \mathbb{E}[Y(1)|X] + k_1 (X, U)   \} | Z_2, X, U \} \big| X \Big] \quad \text{(by the law of total expectation)} \\
    & = \mathbb{E} \Big[  \frac{1}{\mathbb{E}[Z_1 | Z_2, X]} Z_1 \{ \mathbb{E}[Y(1)|X] + k_1 (X, U)   \} \big| X \Big] \\
    & = \mathbb{E}[Y(1) | X] + \mathbb{E} \Big[ \frac{1}{\mathbb{E}[Z_1|Z_2, X] } Z_1 k_1 (X,U) \big| X \Big] \\
    & = \mathbb{E}[Y(1)|X] + \mathbb{E} \Big\{ \sum_{a \in \{0,1 \} } \sum_{b \in \{0,1\}} \frac{a}{\mathbb{E}[Z_1|Z_2 = b, X] } \mathbb{E}[k_1(X,U)|Z_1 = a, Z_2 = b, X]  \mathbb{P}(Z_1 = a, Z_2 = b | X) \mid X \Big\} \\
    & = \mathbb{E}[Y(1)|X] + \mathbb{E} \Big\{  \sum_{b \in \{0,1\}} \frac{1}{\mathbb{E}[Z_1|Z_2 = b, X] } \mathbb{E}[k_1(X,U)|Z_1 = 1, Z_2 = b, X]  \mathbb{P}(Z_1 = 1, Z_2 = b | X) \mid X \Big\} \\
        & = \mathbb{E}[Y(1)|X] + \mathbb{E} \Big\{  \sum_{b \in \{0,1\}} \frac{1}{\mathbb{E}[Z_1|Z_2 = b, X] } \mathbb{E}[k_1(X,U)|Z_1 = 1, Z_2 = b, X]  \mathbb{P}(Z_1 = 1 \mid Z_2 = b,  X) \mathbb{P} (Z_2 = b \mid X) \mid X \Big\} \\
    & = \mathbb{E}[Y(1)|X] +  \sum_{b \in \{0,1\}} \mathbb{E} \Big\{  \mathbb{E}[k_1(X,U)|Z_1 = 1, Z_2 = b, X]  \mathbb{P}(Z_2 = b \mid X) \mid X \Big\}. 
\end{align*}
Similarly, we have
\begin{align*}
    \mathbb{E}\left[\frac{(1 -Z_1) Y}{\mathbb{E} [1 - Z_1|Z_2,X] } \bigg| X\right] =  \mathbb{E}[Y(0) | X] +   \sum_{b \in \{0,1\}} \mathbb{E} \Big\{  \mathbb{E}[k_0(X,U)|Z_1 = 0, Z_2 = b, X]  \mathbb{P}(Z_2 = b \mid X) \mid X \Big\}
\end{align*}
Hence, 
\begin{align*}
    & \qquad \tau (x) + \mathbb{E}[k_1(X,U)|Z_1 = 1, Z_2 = 0, X = x] - \mathbb{E}[k_0(X,U)|Z_1 = 0, Z_2 = 1, X = x] \\
& = \mathbb{E}[Y|Z_1 = 1, Z_2 = 0, X = x] -  \mathbb{E}[Y|Z_1 = 0, Z_2 = 1, X = x] = \mu_1 (x),
\end{align*}
and
\begin{align*}
    & \qquad \tau(x) + \sum_{b \in \{0,1\}} \mathbb{E} \Big\{  \left(\mathbb{E}[k_1(X,U)|Z_1 = 1, Z_2 = b, X] -   \mathbb{E}[k_0(X,U)|Z_1 = 0, Z_2 = b, X] \right) \mathbb{P}(Z_2 = b \mid X) 
\mid X \Big\} \\
    &= \mathbb{E}\left[\frac{Z_1 Y}{\mathbb{E} [Z_1|Z_2,X] }\big| X = x\right] - \mathbb{E}\left[\frac{(1 -Z_1) Y}{\mathbb{E} [1 - Z_1|Z_2,X] }\big| X=x\right] = \mu_2(x).
\end{align*}
It is clear to see that $\tau(x) = \mu_1(x) = \mu_2(x)$ if $l(0,0,x) = l(1,0,x) = l(0,1,x) = l(1,1,x)$. $\mu_1(x) \leq \tau(x) \leq \mu_2(x)$ if $l(0,0,x) \leq l(1,0,x) \leq l(0,1,x) \leq l(1,1,x)$. And $\mu_2(x) \leq \tau(x) \leq \mu_1(x)$ if $l(0,0,x) \geq l(1,0,x) \geq l(0,1,x) \geq l(1,1,x)$. For cases 3.) and 4.), Let
\begin{align*}
    \mathbb{E}[Y(z_1) | Z_2 = 1, X, U] + \epsilon_1 =  \mathbb{E}[Y(z_1) | Z_2 = 0, X, U] + \epsilon_2 = \mathbb{E}[Y(z_1) |  X, U].
\end{align*}
Then by similar algebra as before, we have 
\begin{align*}
    & \qquad \tau(x) + \mathbb{E}[k_1(X,U)|Z_1 = 1, Z_2 = 0, X =x] - \mathbb{E}[k_0(X,U)|Z_1 = 0, Z_2 = 1, X =x] \\
& = \mathbb{E}[Y|Z_1 = 1, Z_2 = 0, X = x] -  \mathbb{E}[Y|Z_1 = 0, Z_2 = 1, X = x] + (\epsilon_1 - \epsilon_2) = \mu_1 + (\epsilon_1 - \epsilon_2),
\end{align*}
and
\begin{align*}
    & \qquad \tau(x) + \sum_{b \in \{0,1\}} \{ \mathbb{E}[k_1(X,U)|Z_1 = 1, Z_2 = b, X = x] - \mathbb{E}[f_0(X,U)|Z_1 = 0, Z_2 = b, X = x] \} \mathbb{P}(Z_2 = b) \\
    &= \mathbb{E}\left[\frac{Z_1 Y}{\mathbb{E} [Z_1|Z_2,X] } \Big| X = x\right] - \mathbb{E}\left[\frac{(1 -Z_1) Y}{\mathbb{E} [1 - Z_1|Z_2,X] } \Big| X = x\right]  = \mu_2(x).
\end{align*}
It is clear to see that $\mu_1(x) + (\epsilon_1 - \epsilon_2) \leq \tau(x) \leq \mu_2(x)$ if $l(0,0,x) \leq l(1,0,x) \leq l(0,1,x) \leq l(1,1,x)$ and $\epsilon_1 - \epsilon_2 \leq 0$ (assumption \ref{assump: RER} (b)). And $\mu_2(x) \leq \tau(x) \leq \mu_1(x) + (\epsilon_1 - \epsilon_2)$ if $l(0,0,x) \geq l(1,0,x) \geq l(0,1,x) \geq l(1,1,x)$ and $\epsilon_1 - \epsilon_2 \geq 0$ (assumption \ref{assump: RER} (a)).
\end{proof}

\subsection*{Conditional Monotonicity Assumption Under Logistic Mixed Effect Model}

By equation \eqref{equ:rasch model}, we have

\begin{align*}
 \mathbb{P}(Z_1 = 0, Z_2 = 0 | U = u, X = x) &=   \frac{1 }{1 + \exp\{\alpha_1 u + \beta_1 x \} + \exp\{\alpha_2 u + \beta_2 x \} + \exp\{\alpha_1 u + \beta_1 x + \alpha_2 u + \beta_2 x \}}, \\
\mathbb{P}(Z_1 = 1, Z_2 = 0 | U = u, X = x) &= \frac{\exp \{ \alpha_1 u + \beta_1 x  \} }{1 + \exp\{\alpha_1 u + \beta_1 x \} + \exp\{\alpha_2 u + \beta_2 x \} + \exp\{\alpha_1 u + \beta_1 x + \alpha_2 u + \beta_2 x \}}, \\
       \mathbb{P}(Z_1 = 0, Z_2 = 1 | U = u, X = x) &= \frac{\exp \{   \alpha_2 u + \beta_2 x \} }{1 + \exp\{\alpha_1 u + \beta_1 x \} + \exp\{\alpha_2 u + \beta_2 x \} + \exp\{\alpha_1 u + \beta_1 x + \alpha_2 u + \beta_2 x \}}, \\
      \mathbb{P}(Z_1 = 1, Z_2 = 1 | U = u, X = x) &= \frac{\exp \{ \alpha_1 u + \beta_1 x +  \alpha_2 u + \beta_2 x \} }{1 + \exp\{\alpha_1 u + \beta_1 x \} + \exp\{\alpha_2 u + \beta_2 x \} + \exp\{\alpha_1 u + \beta_1 x + \alpha_2 u + \beta_2 x \}}.
\end{align*}
Since 
\begin{align*}
    & \quad \frac{1 }{1 + \exp\{\alpha_1 u + \beta_1 x \} + \exp\{\alpha_2 u + \beta_2 x \} + \exp\{\alpha_1 u + \beta_1 x + \alpha_2 u + \beta_2 x \}} \\
    & = \frac{1}{(1 + \exp\{\alpha_1 u \}) (1 + \exp\{\alpha_2 u \})  } \frac{1}{(1 + \exp\{\beta_1 x \}) (1 + \exp\{\beta_2 x \})  },
\end{align*}
by some simple calculations, we then have 

\begin{align*}
     h(0, 0, x) &=  \frac{ \int_U f_0 (x, u) q(u)  g (u|x)  du  }{  \int_U q(u)  g (u|x)  du  }\\
     h(1, 0, x) &=  \frac{ \int_U f_1 (x, u) \exp \{   \alpha_1 u  \}  q(u)  g (u|x)  du  }{  \int_U \exp \{   \alpha_1 u  \}  q(u)  g (u|x)  du  }\\
    h(0, 1, x) &= \frac{ \int_U f_0 (x, u) \exp \{   \alpha_2 u  \}  q(u)  g (u|x)  du  }{  \int_U \exp \{   \alpha_2 u  \}  q(u)  g (u|x)  du  } \\
    h(1, 1, x) &= \frac{ \int_U f_1 (x, u) \exp \{\alpha_1 u + \alpha_2 u  \} q(u) g (u|x) du  }{  \int_U \exp \{\alpha_1 u  + \alpha_2 u  \} q(u)  g (u|x) du  } 
\end{align*}
with $q(u) = 1/[(1 + \exp\{\alpha_1 u \}) (1 + \exp\{\alpha_2 u \})]  $, $g(u|x)$ being density functions of $U$ given $X$. Then $h(0,0,x) \leq h(1,0,x) \leq h(0,1,x) \leq h(1,1,x)$ if $\alpha_2 \geq \alpha_1 > 0$, $\mathbb{E}[Y(a)|X = x, U = u] - \mathbb{E}[Y(a)]$ is non-decreasing in $u$ and $a$ given $x$. And $h(0,0,x) \geq h(1,0,x) \geq h(0,1,x) \geq h(1,1,x)$ if $ \alpha_2 \leq \alpha_1 < 0 $ and $\mathbb{E}[Y(a)|X = x, U = u]- \mathbb{E}[Y(a)]$ is non-increasing in $u$ and $a$ given $x$. 

\section*{Appendix IV: Inference on the ATE}

We also have the following assumption regarding the estimated propensity score. 

\begin{assumption}[Propensity Score Model] \label{assump: PSM}
The estimator of the true propensity score model satisfies $\sup_{X \in \mathcal{X}, Z_2 \in \{0,1\}} |\mathbb{P}(Z_1 = 1| Z_2, X) - \hat{\mathbb{P}}_n (Z_1 = 1| Z_2, X)  | = O_p (n^{-1/2}) $. 
\end{assumption}
Assumption \ref{assump: PSM} is usually satisfied by standard parametric estimation methods under mild regularity conditions. For instance, a logit or probit model that relies on a linear index and is estimated through maximum likelihood will meet assumption \ref{assump: PSM} if $\mathcal{X}$ is bounded.

Denote $\sigma_1^2 = Var(\hat \mu_1)$, $\sigma_2^2 = Var( \hat \mu_2 )$. The following empirical sandwich estimators will be used to obtain the asymptotic variances: 
\begin{align*}
    \hat \sigma_1^2 &= \frac{1}{n(n-1)}\sum_{i=1}^n \left( \hat{\mu}_1(X_i) - \hat{\mu}_1 \right)^2, \\
        \hat \sigma_2^2 &= \frac{1}{n^2} \sum_{i=1}^n \Bigg\{ \frac{Z_{i1} Y_i }{\mathbb{\hat P}_n (Z_{i1}=1|Z_{i2},X) }  - \frac{(1 - Z_{i1}) Y_i }{1-\mathbb{\hat P}_n (Z_{i1}=1|Z_{i2},X) } \\ 
    &- \frac{1}{n} \sum_{i'=1}^n \frac{Z_{i'1} Y_{i'} }{ \mathbb{\hat P}_n (Z_{i'1} = 1|Z_{i'2}, X_{i'})     } + \frac{1}{n} \sum_{i'=1}^n \frac{(1 - Z_{i'1}) Y_{i'} }{1 - \mathbb{\hat P}_n (Z_{i'1} = 1|Z_{i'2}, X_{i'})  }     \Bigg\}^2.
\end{align*}

We implement a two-step approach by \citet{romano2014practical} to construct a confidence interval for $\tau$. Specifically, suppose $\mu_1 \leq \tau \leq \mu_2$, then in the first step, we create a confidence region for the moments $ \mathbb{E} (\hat \mu_1 - \tau)$ and $\mathbb{E}(\tau - \hat \mu_2)$, and use this set to provide information such that the moments are non-positive in the second step. We will consider tests of the null hypothesis
\begin{align*}
    H_0: \mathbb{E} (\hat \mu_1 - \tau ) \leq 0, \mathbb{E}(\tau - \hat \mu_2) \leq 0
\end{align*}
for each $\tau \in \mathbb{R}$ that control the probability of Type I error at level $\alpha$. We list the algorithm below tailored for our problem.

\textbf{Algorithm 2. (Confidence Interval for ATE)}
\begin{itemize}
    \item[1.] Create $L$ bootstrap samplers using $\bm W_i = \{Y_i, Z_{i1}, Z_{i2}, \mathbf{X}_{i} \}_{i=1}^n$, and each sampler has size $n$.   Set $W_{1j} = (\hat \mu_1 - \tau)_j$ and $W_{2j} = (\tau - \hat \mu_2)_j$, and each one is a realization from $j$-th sampler.
    \item[2.] Compute the bootstrap quantile $B_n^{-1} (1 - \beta)$ where
    \begin{align*}
        B_n(x) = \mathbb{P} \Big\{ \max \{\frac{\sqrt{n} ( W_{1j} - \bar W_1  )   }{\hat \sigma_{1} },  \frac{\sqrt{n} ( W_{2l} - \bar W_2  )   }{\hat \sigma_{2} }  \} \leq x   \Big\},
    \end{align*}
    with $\bar W_1 = \sum_{j=1}^L W_{1j}/L, \bar W_2 = \sum_{j=1}^L W_{2j}/L$. 
    \item[3.] Compute $M_n (1 - \beta)$ where
    \begin{align*}
        M_n (1 - \beta) = \Big\{ (W_1, W_2 ) \in \mathbb{R}^2, \max \{\frac{\sqrt{n} ( W_{1} - \bar W_1  )   }{\hat \sigma_{1} },  \frac{\sqrt{n} ( W_{2} - \bar W_2  )   }{\hat \sigma_{2} } \} \leq B_n^{-1} (1 - \beta)  \Big\}. 
    \end{align*}
    \item[4.] Compute $\lambda^*_1, \lambda^*_2$ where
    \begin{align*}
        \lambda^*_1 = \min \{ \bar W_1 + \frac{\hat \sigma_{1} B_n^{-1} (1 - \beta) }{\sqrt{n}}, 0  \}, \quad\lambda^*_2 = \min \{ \bar W_2 + \frac{\hat \sigma_{2} B_n^{-1} (1 - \beta) }{\sqrt{n}}, 0  \}. 
    \end{align*}
    \item[5.] Compute $\hat c_n (1 - \alpha + \beta) = J_n^{-1} (1 - \alpha + \beta, \lambda^*_1, \lambda^*_2)$ where
    \begin{align*}
        J_n (x, \lambda_1, \lambda_2) =  \mathbb{P} \Big\{ \max \{\frac{\sqrt{n} (  \bar W_1 - W_{1j} + \lambda_1  )   }{\hat \sigma_{1} },  \frac{\sqrt{n} ( \bar W_2 -  W_{2j} + \lambda_2 )   }{\hat \sigma_{2} }  \} \leq x   \Big\}. 
    \end{align*}
    \item[6.] Compute 
    \begin{align*}
        \phi_n (\tau) = 1 - \mathbf{1} \Big\{  \{ M_n (1 - \beta) \subseteq \mathbb{R}^2   \} \cup  \{ T_n \leq \hat c_n (1 - \alpha + \beta) \}   \Big\} 
    \end{align*}
    where 
    \begin{align*}
        T_n = \max \{\frac{\sqrt{n} \bar W_1   }{\hat \sigma_{1} },  \frac{\sqrt{n} \bar W_2  }{\hat \sigma_{2} }  \}. 
    \end{align*}
    \item[7.] Obtain the confidence region by
    \begin{align*}
        c_n = \{ \tau \in \mathbb{R}: \phi_n (\tau) = 0  \}. 
    \end{align*}
\end{itemize}

The similar strategy can be applied when $\mu_2 \leq \tau \leq \mu_1$. As suggested by \citet{romano2014practical}, we use $\beta = 0.005$ for constructing the confidence region. Using larger values of $\beta$ generally results in reduced average power. Lower values of $\beta$ do not significantly affect average power but require a larger number of bootstrap samples in the first step. In the following corollary, we show that the confidence region obtained from the above algorithm has the correct level as long as the uniform integrability condition holds. 

\begin{corollary}[Coverage of the ATE]
    Let $W_{1j} = (\hat \mu_1 - \tau)_j$ and $W_{2j} = (\tau - \hat \mu_2)_j$ defined in Algorithm 2, and $\tilde W_1 = \mathbb{E} (\hat \mu_1 - \tau), \tilde W_2 = \mathbb{E} (\tau - \hat \mu_2)$. Suppose $\mathbf{P}$ is such that
    \begin{align*}
        \lim_{\lambda \rightarrow \infty} \sup_{P_1 \in \mathbf{P}} \sup_{\tau \in \mathbb{R}} \mathbb{E}_{P_1}  \Bigg\{ \left( \frac{W_{1j} - \tilde W_1}{ \sigma_1}    \right)^2 \mathbf{1} \left\{  \big\lvert \frac{W_{1j} - \tilde W_1}{ \sigma_1}   \big\lvert > \lambda  \right\} \Bigg\} = 0, \\
                \lim_{\lambda \rightarrow \infty} \sup_{P_2 \in \mathbf{P}} \sup_{\tau \in \mathbb{R}} \mathbb{E}_{P_2} \Bigg\{ \left( \frac{W_{2j} - \tilde W_2}{ \sigma_2}    \right)^2 \mathbf{1} \left\{  \big\lvert \frac{W_{2j} - \tilde W_2}{ \sigma_2}   \big\lvert > \lambda  \right\} \Bigg\} = 0. 
    \end{align*}
    Then $c_n$ obtained from Algorithm 2 with $T_n$ satisfies $\lim_{n \rightarrow \infty} \inf \inf_{(P_1, P_2) \in \mathbf{P}^2} \inf_{\tau \in \mathbb{R}} \mathbb{P} (\tau \in c_n) \geq 1 - \alpha  $. 
\end{corollary}

This result is a direct implication from Theorem 3.1 in \citet{romano2014practical}. 

\section*{Appendix V: Asymptotic Properties of $\hat \mu_1 (x)$ and $\hat \mu_2 (x)$}

   To state the asymptotic properties of the proposed semi-parametric estimators, the following assumption regarding bandwidths is needed. Note that our asymptotic results are conditional on one-dimensional observed covariate $X$ (i.e. $l = 1$), which can be extended to the multi-dimensional case in general.

    \begin{assumption}[Bandwidths] \label{assump: Bandwidth}
The bandwidths $h_1, h_2$ and $h_3$ satisfies the following conditions as $n \rightarrow \infty$: (i) $n_{10} h_1^{2s_1 + l} \rightarrow 0$ and $n_{10} h_1^l \rightarrow \infty$; (ii) $n_{01} h_2^{2s_2 + l} \rightarrow 0$ and $n_{01} h_2^l \rightarrow \infty$; (iii) $n h_3^{2s_3 + l} \rightarrow 0$ and $n h_3^l \rightarrow \infty$. 
\end{assumption}

Assumption \ref{assump: Bandwidth} ensures the asymptotic negligibility of the bias term for kernel regression. The regularity conditions can be stated below:
\begin{itemize}
    \item[(A.1)] $\int |K_1(u)|^{2+\eta_1}du < \infty, \int |K_2(u)|^{2+\eta_2}du < \infty$ and $\int |K_3(u)|^{2+\eta_3}du < \infty $ for some $\eta_1, \eta_2, \eta_3 > 0$. 
    \item[(A.2)] $h_1 \sim n_{10}^{-1/5}, h_2 \sim n_{01}^{-1/5}$ and $h_3 \sim n^{-1/5}$. 
    \item[(A.3)] $m_{10}(x),m_{01}(x),\mu_1 (x), \mu_2(x), f_{10}(x), f_{01}(x)$ and $f(x)$ are twice differentiable. 
    \item[(A.4)] $x_*$ is a continuity point of $\sigma_1^2(x), \sigma_2^2(x), \mathbb{E}\left[|Y|^{2+\eta_1}|Z_1 = 1, Z_2 = 0, X=x\right], \\ \mathbb{E}\left[|Y|^{2+\eta_1}|Z_1 = 0, Z_2 = 1,X=x\right]  $, and $f(x)$. 
\end{itemize}

Then we have the following lemma of the asymptotic normality of $\hat m_{10} (x), \hat m_{01} (x)$. 

\begin{lemma} \label{lem:m(x)}
 Under regularity conditions A.1, A.2, A.3, and A.4, the following hold for each fixed $x_* \in \mathbb{R}$:
\begin{itemize}
    \item[(a.)] $\sqrt{n h_1^{l}}  \frac{\hat m_{10} (x_*) - m_{10} (x_*)}{V_{10}^{1/2}(x_*)} \overset{d}{\to} \mathcal{N} ( B_{10}(x_*), 1  ) $.
    \item[(b.)] $\sqrt{n h_2^{l}} \frac{\hat m_{01} (x_*) - m_{01} (x_*)}{V_{01}^{1/2}(x_*)}\overset{d}{\to} \mathcal{N} ( B_{01}(x_*), 1  ) $.
\end{itemize}
 where
    \begin{align*}
        V_{10}(x) &= \sigma^2_{10}(x) \|K_1 \|_2^2/f_{10}(x),  \\
        V_{01}(x) &= \sigma^2_{01}(x) \|K_2 \|_2^2/f_{01}(x),  \\
        B_{10}(x) & = \left(\int u^2 K_1(u) du \right) \left( \frac{d^2 m_{10}}{ d x^2}  (x) + 2  \frac{d m_{10}}{ d x}  (x) \frac{df_{10}}{dx}(x)/f_{10}(x) \right) \\
                B_{01}(x) & = \left(\int u^2 K_2(u) du \right) \left( \frac{d^2 m_{01}}{ d x^2}  (x) + 2  \frac{d m_{01}}{ d x}  (x) \frac{df_{10}}{dx}(x)/f_{01}(x) \right)
    \end{align*}
    with $\sigma^2_{10}(x) = \mathbb{E} \Big[ \left(\hat m_{10} (X) - m_{10} (X) \right)^2 |X = x  \Big] $, $\sigma^2_{01}(x) = \mathbb{E} \Big[ \left(\hat m_{01} (X) - m_{01} (X) \right)^2 |X = x  \Big] $, and $\|K \|_2^2 = \int K^2(u)du$.

\end{lemma}

\begin{proof}[Proof of Lemma \ref{lem:m(x)}]
Define $\hat r_{10} = \hat m_{10} (x) \hat f_{10} (x) = (n_{10} h_1)^{-1} \sum_{i \in \mathcal{I}_1} K_1 (\frac{x - X_i}{h_1}) Y_i,  \hat r_{01} = \hat m_{01} (x) \hat f_{01} (x) = n_{01}^{-1} \sum_{i \in \mathcal{I}_2} K_1 (x-X_i) Y_i$. The differences $\hat m_{10} (x) - m_{10}, \hat m_{01} (x) - m_{01}$
and their linearization
\begin{align*}
    &\frac{\hat r_{10}(x) - m_{10}(x) \hat f_{10}(x) }{f_{10}(x)}, \\ 
    &\frac{\hat r_{01}(x) - m_{01}(x) \hat f_{01}(x) }{f_{01}(x)}
\end{align*}
are of order $o_p \left( (n_{10}h_1)^{-1/2} \right)$ and $o_p \left( (n_{01}h_1)^{-1/2} \right)$ which can be shown through writing the differences as
\begin{align*}
      &\left( \frac{\hat r_{10}(x) - m_{10}(x) \hat f_{10}(x) }{f_{10}(x)} \right) \left( \frac{f_{10}(x)}{\hat f_{10}(x)} - 1 \right), \\
           &\left( \frac{\hat r_{01}(x) - m_{01}(x) \hat f_{01}(x) }{f_{01}(x)} \right) \left( \frac{f_{01}(x)}{\hat f_{01}(x)} - 1 \right)
\end{align*}
and combining terms. The bias can then be written as
\begin{align*}
    \frac{\mathbb{E} \hat r_{10}(x) - m_{10} \mathbb{E} \hat f_{10}(x) }{f_{10}(x)} \qquad \text{and} \qquad \frac{\mathbb{E} \hat r_{01}(x) - m_{01} \mathbb{E} \hat f_{01}(x) }{f_{01}(x)}
\end{align*}
which approximately equal to $ m^*_{10} (x) - m_{10} (x)$ and $m^*_{01} (x) - m_{01} (x)$ where $m_{10}^* = \mathbb{E}\hat r_{10} (x)/\mathbb{E} \hat f_{10}(x), m_{01}^* = \mathbb{E}\hat r_{01} (x)/\mathbb{E} \hat f_{01}(x) $. We observe that
\begin{align*}
    m^*_{10} (x) - m_{10} (x)& = (\mathbb{E}\{ K_1 (x-X) \})^{-1} \Bigg\{ \int K_1 (x - u )m_{10}(u) f_{10}(u) du - m_{10}(x) f_{10}(x) \\
    &+ m_{10} (x) f_{10}(x) - m_{10}(x) \int K_1(x-u) f_{10}(u) du  \Bigg\} \\
    &\approx \frac{h_1^2}{2} \int u^2 K_1(u) du \left(f_{10}(x)\right)^{-1} \left( \frac{d^2 (m_{10} f_{10}) }{du^2} (x) - m_{10}(x) \frac{d^2 f}{du}(x) \right) \\
    & = \frac{h_1^2}{2} \int u^2 K_1(u) du \left( \frac{d^2 m_{10}}{du^2} (x) + 2 \frac{d m_{10}}{du} (x) \frac{d f_{10}}{du} (x) / f_{10} (x)     \right) 
\end{align*}
and
\begin{align*}
    m^*_{01} (x) - m_{01} (x)& = (\mathbb{E}\{ K_2 (x-X) \})^{-1} \Bigg\{ \int K_2 (x - u )m_{01}(u) f_{01}(u) du - m_{01}(x) f_{01}(x) \\
    &+ m_{01} (x) f_{01}(x) - m_{01}(x) \int K_2(x-u) f_{01}(u) du  \Bigg\} \\
    &\approx \frac{h_2^2}{2} \int u^2 K_2(u) du \left(f_{01}(x)\right)^{-1} \left( \frac{d^2 (m_{01} f_{01}) }{du^2} (x) - m_{01}(x) \frac{d^2 f}{du}(x) \right) \\
    & = \frac{h_2^2}{2} \int u^2 K_2(u) du \left( \frac{d^2 m_{01}}{du^2} (x) + 2 \frac{d m_{01}}{du} (x) \frac{d f_{01}}{du} (x) / f_{01} (x)     \right). 
\end{align*}
Note that
\begin{align*}
    \hat m_{10} (x) - m_{10}^* (x) &=  \underbrace{\left( \hat r_{10}(x)/f_{10}(x) -  \hat f_{10}(x) m_{10}^*(x) /f_{10}(x)    \right)}_{H_{10} (x) } \left( f_{10}(x)/\hat f_{10}(x)  \right) \\
        \hat m_{01} (x) - m_{01}^* (x) &= \underbrace{\left( \hat r_{01}(x)/f_{01}(x) -  \hat f_{01}(x) m_{01}^*(x) /f_{01}(x)    \right)}_{H_{01}(x)} \left( f_{01}(x)/\hat f_{01}(x)  \right)
\end{align*}
and $f_{10}(x)/\hat f_{10}(x) \overset{p}{\rightarrow} 1, f_{01}(x)/\hat f_{01}(x) \overset{p}{\rightarrow} 1$, so that $\hat m_{10} (x) - m_{10}^* (x), \hat m_{01} (x) - m_{01}^* (x) $ will have the same asymptotic distributions as $H_{10}(x), H_{01}(x)$. It can be shown that 
\begin{align*}
    &n_{10} h_1 Var\left(H_{10} (x) \right) \rightarrow \left(f_{10}(x)  \right)^{-1} \sigma_{10}(x) \int K_1^2(u) du \\
    & n_{10} Cov\left(H_{10} (x), H_{10}(y)  \right) \rightarrow 0,  \text{as } n_{10} \rightarrow \infty. \\
        &n_{01} h_2 Var\left(H_{01} (x) \right) \rightarrow \left(f_{01}(x)  \right)^{-1} \sigma_{01}(x) \int K_2^2(u) du \\
    & n_{01} Cov\left(H_{01} (x), H_{01}(y)  \right) \rightarrow 0,  \text{as } n_{01} \rightarrow \infty. 
\end{align*}
Applying the central limit theorem then yields the asymptotic normality of 
\begin{align*}
    \sqrt{n h_1^{l}} \frac{\hat m_{10} (x_*) - m_{10} (x_*)}{V_{10}^{1/2}(x_*)} \text{ and } \sqrt{n h_2^{l}} \frac{\hat m_{01} (x_*) - m_{01} (x_*)}{V_{01}^{1/2}(x_*)}. 
\end{align*}
\end{proof}

We then have our asymptotic property for $\hat \mu_1 (x)$ and $\hat \mu_2 (x)$.
 
\begin{theorem} \label{thm: inference tau(x)}
Assume that $h_1 \sim h_2$. Suppose (\ref{def: bound tau hat (x)}) are satisfied for some $s_1, s_2, s_3 \geq 2$.  Under assumption \ref{assump: PSM}, \ref{assump: Bandwidth}, and additional regularity conditions A.1, A.2, A.3, A.4, the following hold for each fixed $x_* \in \mathbb{R}$:

\begin{itemize}
    \item[(1.)] $\sqrt{n } \frac{\hat \mu_1 (x_*) - \mu_1 (x_*)}{\sigma_1 (x_*)}\overset{d}{\to} \mathcal{N} \left( B(x_*)  , 1   \right) $,
    \item[(2.)] $\sqrt{n h_3^{l}} \left( \hat \mu_2 (x_*) - \mu_2 (x_*) \right)  = \frac{1}{\sqrt{n h_3^l}} \frac{1}{f(x_*)} \sum_{i=1}^n \left(\phi(X_i,Y_i,Z_{i1},Z_{i2}) - \mu_2(x_*)  \right) K_3 (\frac{X_i - x_*}{h_3} ) + o_p(1) $,
    \item[(3.)] $\sqrt{n h_3^{l}} \frac{ \hat \mu_2 (x_*) - \mu_2 (x_*)}{\sigma_2 (x_*)} \overset{d}{\to} \mathcal{N} \left(0, 1 \right) $,
\end{itemize}
 where 
 
 $\sigma^2_1 (x) = V_{10} (x) + V_{01} (x), \sigma^2_2(x) = \sigma^2_\phi(x) \|K_3 \|_2^2/f(x), B(x) = B_{10}(x) - B_{01}(x) $. 

\end{theorem}

\begin{proof}[Proof of Theorem \ref{thm: inference tau(x)}]
(1) is proved through Lemma \ref{lem:m(x)} and the subtraction of two estimators. For (2), let $w = (y,z_1,z_2,x)$ and $\Phi(w, p) = \frac{z_1y}{p} - \frac{(1-z_1)y}{1-p}$ , we first expand the numerator of $\hat \mu_2 (x)$ around $p(X_i,Z_{i2})$ as
\begin{align*}
   &\quad \frac{1}{\sqrt{n h_3}} \sum_{i=1}^n K_3 (\frac{X_{i1} - x }{h_3}  ) \left( \Phi (W_i, \hat p(X_i, Z_{i1}) ) - \mu_2 (x)     \right) \\
    & = \frac{1}{\sqrt{n h_3}} \sum_{i=1}^n K_3 (\frac{X_{i1} - x }{h_3}  ) \left( \Phi (W_i,  p(X_i, Z_{i1}) ) - \mu_2 (x)     \right) + \\ &\frac{1}{\sqrt{n h_3}} \sum_{i=1}^n K_3 (\frac{X_{i1} - x }{h_3}  ) \Phi (W_i,  p^*(X_i, Z_{i1}) ) \left(  \hat p(X_i, Z_{i1})  - p(X_i, Z_{i1}) )     \right)
\end{align*}
where $p^*(X_i, Z_{i1})$ is between $p(X_i, Z_{i1})$ and $\hat p(X_i, Z_{i1})$. We bound the second term as
\begin{align*}
 &\quad    \frac{1}{\sqrt{n h_3}} \Bigg|  \sum_{i=1}^n K_3 (\frac{X_{i1} - x }{h_3}  ) \Phi (W_i,  p^*(X_i, Z_{i1}) ) \left(  \hat p(X_i, Z_{i1})  - p(X_i, Z_{i1}) )     \right)  \Bigg| \\
    & \leq \sqrt{n h_3} \sup_{x \in \mathbb{R}, z_2 \in \mathbb{R}} |\hat p(x,z_2) - p(x,z_2) | \cdot  \frac{1}{\sqrt{n h_3}}  \sum_{i=1}^n K_3 (\frac{X_{i1} - x }{h_3}  ) \Phi (W_i,  p^*(X_i, Z_{i1}) ),
\end{align*}
where the first factor is $o_p(1)$ by assumption \ref{assump: PSM} and the second factor is $O_p(1)$ by regularity conditions A.1, A.2, A.3, and A.4. Since $ \Phi (W_i,  p(X_i, Z_{i1}) ) = \phi(X_i,Y_i,Z_{i1},Z_{i2})$, (2) is proved. For (3), we first write
\begin{align*}
& \quad \sqrt{n h_3} (\hat \mu_2(x_*) - \mu_2(x_*) ) \\
    & = \frac{1}{\sqrt{n h_3}} \frac{1}{f(x_*)} \sum_{i=1}^n \left( \phi(X_i, Y_i, Z_{i1}, Z_{i2}) - \mu_2(X_i^*)  \right) K_3 (\frac{X_i^* - x}{h_3}) \\
    & + \frac{1}{\sqrt{n h_3}} \frac{1}{f(x_*)} \sum_{i=1}^n \left(  \mu_2(X_i^*) -  \mu_2(x_*)   \right) K_3(\frac{X_i^* - x_*}{h_3}) + o_p(1).
\end{align*}
It is easy to show that $\mathbb{E} \{ \left(\phi(X_i, Y_i, Z_{i1}, Z_{i2}) - \mu_2(X_i^*) \right) K_{3in} \} = 0 $ where $K_{3in} = K_3(\frac{X_i^*-x_*}{h_3})$. For each $n$, the random variables $ \{\left(\phi(X_i, Y_i, Z_{i1}, Z_{i2}) - \mu_2(X_i^*) \right) K_{3in}  \}_{i=1}^n$ are independent and one can apply Lyapunov's CLT for triangular arrays to obtain the asymptotic distribution shown in (3). The verification of the conditions of Lyapunov's CLT mimics exactly the proof of Theorem 3.5 by \citet{ullah1999nonparametric}. And the bias term is $o_p(1)$ under assumption \ref{assump: Bandwidth} and therefore has no bearing on the limit distribution. 
\end{proof}

\section*{Appendix VI: Additional Simulation Results and Case Studies Results}

\begin{table}[!htb] 
\centering
\begin{tabular}{  cccccccc  p{5cm} }
\hline
 &  Outcome model & $p = 0.7$ & $p = 0.8$ & $p = 0.9$ & $p = 0.7$ & $p = 0.8$ & $p = 0.9$ \\ 
 \hline
   \multicolumn{7}{c}{\qquad \qquad \qquad \qquad \qquad \qquad \qquad $n = 1000,\delta = 0$ \quad \quad  \qquad \qquad  $n = 1000,\delta = 1$}   \\ 
  \multirow{2}{4em}{$d = 5$} & \emph{Homogeneous} & 0.954 & 0.914 & 0.820 & 0.978 & 0.954 & 0.864 \\ 
   & \emph{Heterogeneous}  & 0.978 & 0.956 & 0.842 & 0.988 & 0.954 & 0.892 \\ 
  \multirow{2}{4em}{$d = 10$} &\emph{Homogeneous}  & 0.996 & 0.986 & 0.926 & 1.000 & 0.984 & 0.918 \\ 
   & \emph{Heterogeneous}  & 0.992 & 0.976 & 0.920 & 0.998 & 0.986 & 0.918 \\ 
 \multirow{2}{4em}{$d = 20$}  & \emph{Homogeneous}  & 0.998 & 0.996 & 0.944 & 1.000 & 0.994 & 0.952 \\ 
   & \emph{Heterogeneous} & 0.996 & 0.998 & 0.958 & 1.000 & 0.996 & 0.956 \\ 
   \hline
   \multicolumn{7}{c}{\qquad \qquad \qquad \qquad \qquad \qquad \qquad $n = 2000,\delta = 0$ \quad \quad  \qquad \qquad  $n = 2000,\delta = 1$}   \\ 
\multirow{2}{4em}{$d = 5$} & \emph{Homogeneous} & 0.998 & 0.994 & 0.928 & 0.996 & 0.992 & 0.946 \\ 
   & \emph{Heterogeneous} & 1.000 & 0.994 & 0.956 & 1.000 & 0.998 & 0.970 \\ 
  \multirow{2}{4em}{$d = 10$} & \emph{Homogeneous} & 1.000 & 0.996 & 0.958 & 1.000 & 0.998 & 0.988 \\ 
   & \emph{Heterogeneous} & 1.000 & 0.998 & 0.974 & 1.000 & 0.998 & 0.984 \\ 
  \multirow{2}{4em}{$d = 20$} & \emph{Homogeneous} & 1.000 & 1.000 & 0.992 & 1.000 & 1.000 & 0.992 \\ 
   & \emph{Heterogeneous} & 1.000 & 1.000 & 0.994 & 1.000 & 1.000 & 0.988 \\ 
   \hline
   \multicolumn{7}{c}{\qquad \qquad \qquad \qquad \qquad \qquad \qquad $n = 5000,\delta = 0$ \quad \quad  \qquad \qquad  $n = 5000,\delta = 1$}   \\ 
  \multirow{2}{4em}{$d = 5$}  & \emph{Homogeneous} & 1.000 & 1.000 & 0.994 & 1.000 & 1.000 & 0.998 \\ 
    & \emph{Heterogeneous}  & 1.000 & 1.000 & 0.998 & 1.000 & 1.000 & 0.996 \\ 
  \multirow{2}{4em}{$d = 10$}  & \emph{Homogeneous}  & 1.000 & 1.000 & 1.000 & 1.000 & 1.000 & 1.000 \\ 
    & \emph{Heterogeneous} & 1.000 & 1.000 & 1.000 & 1.000 & 1.000 & 1.000 \\ 
  \multirow{2}{4em}{$d = 20$}  & \emph{Homogeneous}& 1.000 & 1.000 & 1.000 & 1.000 & 1.000 & 1.000 \\ 
    & \emph{Heterogeneous} & 1.000 & 1.000 & 1.000 & 1.000 & 1.000 & 0.996 \\ 
    \hline
\end{tabular}
\caption{Reported coverage probability of the estimated bounds for the ATE across different combinations of factors.}
\label{tab: ATE}
\end{table}

\begin{table}[!htb] 
\centering
\begin{tabular}{  cccccccc  p{5cm} }
\hline
 &  Outcome model & $p = 0.7$ & $p = 0.8$ & $p = 0.9$ & $p = 0.7$ & $p = 0.8$ & $p = 0.9$ \\ 
 \hline
   \multicolumn{7}{c}{\qquad \qquad \qquad \qquad \qquad \qquad \qquad $n = 1000,\delta = 0$ \quad \quad  \qquad \qquad  $n = 1000,\delta = 1$}   \\ 
  \multirow{2}{4em}{$d = 5$} & \emph{Homogeneous} & 0.718 & 0.692 & 0.590 & 0.756 & 0.772 & 0.670 \\ 
   & \emph{Heterogeneous} & 0.756 & 0.740 & 0.666 & 0.788 & 0.752 & 0.696 \\ 
  \multirow{2}{4em}{$d = 10$} &\emph{Homogeneous} & 0.802 & 0.828 & 0.786 & 0.818 & 0.824 & 0.738 \\ 
   & \emph{Heterogeneous}  & 0.836 & 0.830 & 0.748 & 0.836 & 0.844 & 0.746 \\  
 \multirow{2}{4em}{$d = 20$}  & \emph{Homogeneous} & 0.900 & 0.900 & 0.840 & 0.890 & 0.886 & 0.828 \\ 
   & \emph{Heterogeneous} & 0.900 & 0.934 & 0.866 & 0.886 & 0.886 & 0.816 \\ 
   \hline
   \multicolumn{7}{c}{\qquad \qquad \qquad \qquad \qquad \qquad \qquad $n = 2000,\delta = 0$ \quad \quad  \qquad \qquad  $n = 2000,\delta = 1$}   \\ 
\multirow{2}{4em}{$d = 5$} & \emph{Homogeneous} & 0.846 & 0.852 & 0.740 & 0.864 & 0.854 & 0.808 \\ 
   & \emph{Heterogeneous} & 0.866 & 0.860 & 0.782 & 0.860 & 0.868 & 0.838 \\ 
  \multirow{2}{4em}{$d = 10$} & \emph{Homogeneous}& 0.914 & 0.924 & 0.872 & 0.924 & 0.932 & 0.896 \\ 
   & \emph{Heterogeneous}& 0.936 & 0.940 & 0.902 & 0.912 & 0.908 & 0.890 \\ 
  \multirow{2}{4em}{$d = 20$} & \emph{Homogeneous} & 0.960 & 0.958 & 0.932 & 0.958 & 0.962 & 0.938 \\ 
   & \emph{Heterogeneous} & 0.968 & 0.976 & 0.952 & 0.960 & 0.968 & 0.926 \\ 
   \hline
   \multicolumn{7}{c}{\qquad \qquad \qquad \qquad \qquad \qquad \qquad $n = 5000,\delta = 0$ 
   \quad \quad  \qquad \qquad  $n = 5000,\delta = 1$}   \\ 
  \multirow{2}{4em}{$d = 5$}  & \emph{Homogeneous}  & 0.944 & 0.926 & 0.892 & 0.926 & 0.922 & 0.916 \\ 
    & \emph{Heterogeneous}  & 0.936 & 0.934 & 0.916 & 0.926 & 0.932 & 0.916 \\ 
  \multirow{2}{4em}{$d = 10$}  & \emph{Homogeneous} & 0.972 & 0.978 & 0.972 & 0.974 & 0.986 & 0.958 \\ 
    & \emph{Heterogeneous} & 0.988 & 0.982 & 0.968 & 0.980 & 0.992 & 0.964 \\ 
  \multirow{2}{4em}{$d = 20$}  & \emph{Homogeneous}& 0.984 & 0.996 & 0.992 & 0.990 & 0.998 & 0.990 \\ 
    & \emph{Heterogeneous} & 0.990 & 0.998 & 0.990 & 0.988 & 0.998 & 0.992 \\ 
    \hline
\end{tabular}
\caption{Reported coverage probability of the estimated bounds for the CATE at $X_1 = 1$ across different combinations of factors.}
\label{tab: CATE}
\end{table}

\begin{table}[!htb] 
\centering
\begin{tabular}{  cccccccc  p{5cm} }
\hline
 &  Outcome model & $p = 0.7$ & $p = 0.8$ & $p = 0.9$ & $p = 0.7$ & $p = 0.8$ & $p = 0.9$ \\ 
 \hline
   \multicolumn{7}{c}{\qquad \qquad \qquad \qquad \qquad \qquad \qquad $n = 1000,\delta = 0$ \quad \quad  \qquad \qquad  $n = 1000,\delta = 1$}   \\ 
  \multirow{2}{4em}{$d = 5$} & \emph{Homogeneous} & 1.461 & 1.417 & 1.337 & 1.413 & 1.238 & 1.276 \\ 
   & \emph{Heterogeneous}  & 1.459 & 1.376 & 1.462 & 1.373 & 1.416 & 1.211 \\ 
  \multirow{2}{4em}{$d = 10$} &\emph{Homogeneous}  & 1.553 & 1.372 & 1.415 & 1.463 & 1.463 & 1.323 \\ 
   & \emph{Heterogeneous}  & 1.482 & 1.409 & 1.318 & 1.462 & 1.405 & 1.288 \\ 
 \multirow{2}{4em}{$d = 20$}  & \emph{Homogeneous}  & 1.451 & 1.377 & 1.233 & 1.58 & 1.358 & 1.383 \\ 
   & \emph{Heterogeneous} & 1.513 & 1.297 & 1.384 & 1.644 & 1.386 & 1.312 \\ 
   \hline
   \multicolumn{7}{c}{\qquad \qquad \qquad \qquad \qquad \qquad \qquad $n = 2000,\delta = 0$ \quad \quad  \qquad \qquad  $n = 2000,\delta = 1$}   \\ 
\multirow{2}{4em}{$d = 5$} & \emph{Homogeneous} & 1.514 & 1.445 & 1.158 & 1.465 & 1.376 & 1.217 \\ 
   & \emph{Heterogeneous} & 1.617 & 1.368 & 1.122 & 1.618 & 1.392 & 1.205 \\ 
  \multirow{2}{4em}{$d = 10$} & \emph{Homogeneous} & 1.443 & 1.331 & 1.347 & 1.580 & 1.423 & 1.283 \\ 
   & \emph{Heterogeneous} & 1.485 & 1.385 & 1.422 & 1.569 & 1.223 & 1.385 \\ 
  \multirow{2}{4em}{$d = 20$} & \emph{Homogeneous} & 1.517 & 1.581 & 1.232 & 1.474 & 1.317 & 1.282 \\ 
   & \emph{Heterogeneous} & 1.403 & 1.423 & 1.427 & 1.488 & 1.401 & 1.246 \\ 
   \hline
   \multicolumn{7}{c}{\qquad \qquad \qquad \qquad \qquad \qquad \qquad $n = 5000,\delta = 0$ \quad \quad  \qquad \qquad  $n = 5000,\delta = 1$}   \\ 
  \multirow{2}{4em}{$d = 5$}  & \emph{Homogeneous} & 1.552 & 1.335 & 1.271 & 1.436 & 1.314 & 1.243 \\ 
    & \emph{Heterogeneous}  & 1.538 & 1.471 & 1.207 & 1.631 & 1.268 & 1.537 \\ 
  \multirow{2}{4em}{$d = 10$}  & \emph{Homogeneous}  & 1.496 & 1.436 & 1.266 & 1.452 & 1.534 & 1.253 \\ 
    & \emph{Heterogeneous} & 1.535 & 1.321 & 1.384 & 1.584 & 1.371 & 1.288 \\ 
  \multirow{2}{4em}{$d = 20$}  & \emph{Homogeneous}& 1.491 & 1.362 & 1.319 & 1.516 & 1.442 & 1.334 \\ 
    & \emph{Heterogeneous} & 1.446 & 1.438 & 1.261 & 1.459 & 1.291 & 1.333 \\ 
    \hline
\end{tabular}
\caption{Average relative length of the $95\%$ confidence interval for the ATE across different combinations of factors. Relative Length = $(Upper_{C.I.} - Lower_{C.I.})/(Upper - Lower)$ where $Upper_{C.I.}, Lower_{C.I.}$ represent upper and lower confidence interval while $Upper, Lower$ represent upper bound and lower bound.  }
\label{tab: relative length CI ATE}
\end{table}

\begin{table}[!htb] 
\centering
\begin{tabular}{  cccccccc  p{5cm} }
\hline
 &  Outcome model & $p = 0.7$ & $p = 0.8$ & $p = 0.9$ & $p = 0.7$ & $p = 0.8$ & $p = 0.9$ \\ 
 \hline
   \multicolumn{7}{c}{\qquad \qquad \qquad \qquad \qquad \qquad \qquad $n = 1000,\delta = 0$ \quad \quad  \qquad \qquad  $n = 1000,\delta = 1$}   \\ 
  \multirow{2}{4em}{$d = 5$} & \emph{Homogeneous} & 1.596 & 1.626 & 1.457 & 1.637 & 1.537 & 1.338 \\ 
   & \emph{Heterogeneous}  & 1.658 & 1.532 & 1.439 & 1.637 & 1.471 & 1.275 \\ 
  \multirow{2}{4em}{$d = 10$} &\emph{Homogeneous}  & 1.699 & 1.624 & 1.419 & 1.678 & 1.521 & 1.304 \\ 
   & \emph{Heterogeneous}  & 1.635 & 1.495 & 1.432 & 1.789 & 1.512 & 1.465 \\ 
 \multirow{2}{4em}{$d = 20$}  & \emph{Homogeneous}  & 1.488 & 1.487 & 1.433 & 1.575 & 1.427 & 1.344 \\ 
   & \emph{Heterogeneous} & 1.477 & 1.382 & 1.391 & 1.619 & 1.516 & 1.382 \\ 
   \hline
   \multicolumn{7}{c}{\qquad \qquad \qquad \qquad \qquad \qquad \qquad $n = 2000,\delta = 0$ \quad \quad  \qquad \qquad  $n = 2000,\delta = 1$}   \\ 
\multirow{2}{4em}{$d = 5$} & \emph{Homogeneous} & 1.569 & 1.456 & 1.358 & 1.662 & 1.677 & 1.419 \\ 
   & \emph{Heterogeneous} & 1.712 & 1.483 & 1.323 & 1.713 & 1.426 & 1.424 \\ 
  \multirow{2}{4em}{$d = 10$} & \emph{Homogeneous} & 1.648 & 1.404 & 1.286 & 1.603 & 1.262 & 1.392\\ 
   & \emph{Heterogeneous} & 1.396 & 1.497 & 1.402 & 1.617 & 1.438 & 1.472 \\ 
  \multirow{2}{4em}{$d = 20$} & \emph{Homogeneous} & 1.639 & 1.458 & 1.416 & 1.522 & 1.561 & 1.377 \\ 
   & \emph{Heterogeneous} & 1.598 & 1.501 & 1.407 & 1.609 & 1.494 & 1.375 \\ 
   \hline
   \multicolumn{7}{c}{\qquad \qquad \qquad \qquad \qquad \qquad \qquad $n = 5000,\delta = 0$ \quad \quad  \qquad \qquad  $n = 5000,\delta = 1$}   \\ 
  \multirow{2}{4em}{$d = 5$}  & \emph{Homogeneous} & 1.569 & 1.453 & 1.359 & 1.551 & 1.408 & 1.268 \\ 
    & \emph{Heterogeneous}  & 1.587 & 1.485 & 1.399 & 1.623 & 1.571 & 1.475 \\ 
  \multirow{2}{4em}{$d = 10$}  & \emph{Homogeneous}  & 1.581 & 1.536 & 1.385 & 1.556 & 1.538 & 1.457 \\ 
    & \emph{Heterogeneous} & 1.539 & 1.438 & 1.357 & 1.729 & 1.623 & 1.347 \\ 
  \multirow{2}{4em}{$d = 20$}  & \emph{Homogeneous}& 1.536 & 1.465 & 1.362 & 1.592 & 1.508 & 1.531 \\ 
    & \emph{Heterogeneous} & 1.631 & 1.473 & 1.366 & 1.557 & 1.534 & 1.475 \\ 
    \hline
\end{tabular}
\caption{Average relative length of the $95\%$ confidence interval for the CATE at $X_1 = 1$ across different combinations of factors. Relative Length = $(Upper_{C.I.} - Lower_{C.I.})/(Upper - Lower)$ where $Upper_{C.I.}, Lower_{C.I.}$ represent upper and lower confidence interval while $Upper, Lower$ represent upper bound and lower bound.  }
\label{tab: relative length CI CATE}
\end{table}

\begin{table}[!htb]
\centering
\caption{Demographic distribution of data example stratified by current smoking status from NHANES 2017-2018.}
\label{tab:demographics NHANES}
\resizebox{1\textwidth}{!}{
\begin{tabular}{lccc}
\hline
 \multirow{3}{*}{\begin{tabular}{c}\textbf{Characteristic}\end{tabular}}
   & \multirow{3}{*}{\begin{tabular}{c} Non-smoker \\  ($n = 1908$)\end{tabular}} &
   \multirow{3}{*}{\begin{tabular}{c}Current smoker \\ ($n = 381$)\end{tabular}} &
 \\ \\ \\
\hline
    \textbf{Age at enrollment} & 45(32,57) & 46(35,57) 	\\
    \textbf{Gender, No.(\%)} \\
  \hspace{0.3cm} Male & 786(41) & 226(59) \\ 
  \hspace{0.3cm} Female & 1122(59)	 & 155(41)   \\ 
    \textbf{Blood level of lead (umol/L)} & 0.04(0.02,0.06) & 0.06(0.04,0.09)  \\
   \textbf{past used cocaine/heroin/methamphetamine, No.(\%)} \\
  \hspace{0.3cm} No & 1780(93.3) & 232(61)  \\ 
  \hspace{0.3cm} Yes & 128(6.7)	 & 149(39)   \\ 
  \textbf{Blood level of cadmium (nmol/L)} & 2(1.3,3.3) & 10.1(6.8,14.9)  \\
  \textbf{Ethnicity No.(\%)} \\
  \hspace{0.3cm} Mexican American & 310(16) & 12(3.1)  \\ 
  \hspace{0.3cm} Other Hispanic & 197(10) & 16(4.2)  \\ 
  \hspace{0.3cm} Non-Hispanic White & 508(27) & 217(57)  \\ 
  \hspace{0.3cm} Non-Hispanic Black & 461(24) & 77(20)  \\ 
  \hspace{0.3cm} Non-Hispanic Asian & 351(18) & 18(4.7)  \\ 
  \hspace{0.3cm} Other Race - Including Multi-Racial & 81(4.2) & 41(11)  \\ 
    \textbf{Ratio of family income to poverty} & 2.64(1.36,4.80) & 1.27(0.81,2.15)  \\
  \textbf{Education, No.(\%)} \\
  \hspace{0.3cm} Less than 9th grade & 116(6.1) & 12(3.1)  \\ 
  \hspace{0.3cm} 9-11th grade (Includes 12th grade with no diploma) & 146(7.7) & 77(20)  \\ 
  \hspace{0.3cm} High school graduate/GED or equivalent & 391(20) & 137(36)  \\ 
  \hspace{0.3cm} Some college or AA degree & 629(33) & 132(35) \\ 
  \hspace{0.3cm} College graduate or above & 625(33) & 23(6)  \\ 
\hline
\end{tabular}}
\end{table}

\end{document}